\documentclass[a4paper,11pt]{article}                    
\usepackage{vmargin}
\setmarginsrb{2.4cm}{2.4cm}{2.4cm}{2.4cm}{0pt}{5mm}{0pt}{6mm}
\usepackage{graphicx}
\usepackage{amssymb,amsmath,textcomp}
\usepackage{amsthm}
\usepackage{color}
\usepackage{hyperref}
\usepackage{tabularx}
\usepackage[justification=centering]{caption}
\usepackage{latexsym}
\usepackage{authblk}

\newcommand{\remove}[1]{}

\newcommand{\defparproblem}[4]{
 \vspace{3mm}
\noindent\fbox{
  \begin{minipage}{.95\textwidth}
  \begin{tabular*}{\textwidth}{@{\extracolsep{\fill}}lr} \textsc{#1} \\ \end{tabular*}
  {\bf{Input:}} #2  \\
  {\bf{Parameter:}} #3\\
  {\bf{Question:}} #4
  \end{minipage}
  }
  \vspace{2mm}
}

\usepackage{microtype}

\newtheorem{definition}{\bf Definition}[section]
\newtheorem{observation}{Observation}
\newtheorem{reduction rule}{\bf Reduction Rule}[section]
\newtheorem{lemma}{\bf Lemma}[section]
\newtheorem{theorem}{\bf Theorem}[section]

\newtheorem{corollary}{Corollary}[section]
\newtheorem{conjecture}{Conjecture}[section]
\newtheorem{branching rule}{Branching Rule}[section]

\newtheorem{prop}{Proposition}[section]

\newcommand{\BB}{{\mathcal B}}

\newcommand{\FF}{\ensuremath{\mathcal{F}}\xspace}

\newcommand{\OO}{\mathcal{O}}

\newcommand{\nka}{${\sf NP \subseteq coNP/poly}$}
\newcommand{\FPT}{{\sf FPT}}
\newcommand{\NP}{{\sf NP}}
\newcommand{\fP}{{\sf P}}
\newcommand{\W}{{\sf W}}
\newcommand{\XP}{\sf XP}

\newcommand{\EDS}{{\sc Efficient Dominating Set }}
\newcommand{\SCWP}{{\sc Set-Cover with Partition}}
\newcommand{\ESCWP}{{\sc Exact Set-Cover with Partition}}
\newcommand{\WSMCWP}{{\sc Weighted Set-Multicover with Partition}}
\newcommand{\SetCover}{{\sc Set-Cover}}

\usepackage{tikz}
\usetikzlibrary{decorations.shapes,decorations.pathreplacing,decorations.pathmorphing}
\usetikzlibrary{arrows,matrix,shapes}
\usetikzlibrary{positioning}
\tikzset{
        stars/.style={star,inner sep=2pt}
    }

\usepackage{todonotes}
\usepackage[ruled,vlined,linesnumbered]{algorithm2e}

\title{Parameterized Complexity of Dominating Set Variants in Almost Cluster and Split Graphs\footnote{Some of the results from this paper appeared in proceedings of 13th International Computer Symposium in Russia (CSR) 2018~\cite{GJKMR18csr}.}}

\date{}

\author[1]{Dishant Goyal}
\author[2]{Ashwin Jacob}
\author[3]{Kaushtubh Kumar}
\author[4]{Diptapriyo Majumdar}
\author[5]{Venkatesh Raman}
\affil[1]{Indian Institute of Technology Delhi, New Delhi, India\\
\texttt{{dishant.in@gmail.com}}}
\affil[2]{National Institute of Technology Calicut, India\\
  \texttt{ashwinjacob@nitc.ac.in}}
\affil[3]{Mentor Graphics Corporation, Noida, India\\
\texttt{{kaushtubhkp10@gmail.com}}}
\affil[4]{Indraprastha Institute of Information Technology Delhi, New Delhi, India\\
\texttt{diptapriyo@iiitd.ac.in}}
\affil[5]{The Institute of Mathematical Sciences, HBNI, Chennai, India\\
\texttt{vraman@imsc.res.in}}

\begin{document}
\maketitle

\begin{abstract}
We consider structural parameterizations of the fundamental {\sc Dominating Set} problem and its variants in the parameter ecology program.
We give improved FPT algorithms and lower bounds under well-known conjectures for dominating set in graphs that are $k$ vertices away from a cluster graph or a split graph.
These are graphs in which there is a set of $k$ vertices (called the modulator) whose deletion results in a cluster graph or a split graph.
We also call $k$ as the deletion distance (to the appropriate class of graphs).
Specifically, we show the following results.
When parameterized by the deletion distance $k$ to cluster graphs: 
\begin{itemize}
\item
we can find a minimum dominating set (DS) in $3^k n^{\OO(1)}$-time.
Within the same time, we can also find a minimum independent dominating set (IDS) or a minimum dominating clique (DC) or a minimum efficient dominating set (EDS) or a minimum total dominating set (TDS).
 These algorithms are obtained through a dynamic programming approach for an interesting generalization of set cover which may be of independent interest.
\item
We complement our upper bound results by showing that at least for minimum DS, DC and TDS, $(2-\epsilon)^k n^{\OO(1)}$-time algorithm is not possible for any $\epsilon > 0$ under, what is known as, Set Cover Conjecture. 
We also show that most of these variants of dominating set do not have polynomial sized kernel. 
\end{itemize}
The standard dominating set and most of its variants are \NP-hard or \W[2]-hard in split graphs. For the two variants IDS and EDS that are polynomial time solvable in split graphs, we show that when parameterized by the deletion distance $k$ to split graphs,
\begin{itemize}
\item
IDS can be solved in $2^k n^{\OO(1)}$ time and we provide an $\Omega(2^k)$ lower bound under the strong exponential time hypothesis (SETH);
\item
the $2^k$ barrier can be broken for EDS by designing an  $3^{k/2}n^{\OO(1)}$ algorithm. This is one of the very few problems with a runtime better than $2^k n^{\OO(1)}$ in the realm of structural parameterization. We also show that no $2^{o(k)}n^{\OO(1)}$ algorithm is possible unless the exponential time hypothesis (ETH) is false.
\end{itemize}
\end{abstract}

\section{Introduction}
\label{sec:intro}
\subsection{Motivation}
The {\sc Dominating Set} problem is one of the classical \NP-Complete graph theoretic problems \cite{garey1979computers}. 
It asks for a minimum set of vertices in a graph such that every vertex is either in that set or has a neighbor in that set. It, along with several variations, including
{\it independent domination, total domination, efficient domination, connected domination, total perfect domination, threshold domination} are well-studied in all algorithmic paradigms \cite{garey1979computers} , including parameterized complexity \cite{downey1995fixed,alber2004polynomial,philip2012polynomial}, exact algorithms \cite{van2011exact}, approximation algorithms \cite{downey2008parameterized,ChlebikC06,ChlebikC08}, and structural \cite{haynes1997domination,haynes2023fundamentals} points of view.
All of these versions are hard for the parameterized complexity class \W[1] in general graphs when parameterized by solution size \cite{downey1995fixed} and hence are unlikely to be fixed-parameter tractable (See~\cite{CFKLMPPS15} for more details).

One of the goals of parameterized complexity is to identify parameters under which the NP-hard problems are fixed-parameter tractable. This is also of practical interest as often there are some small parameters (other than solution size) that capture important practical inputs. 
This has resulted in the parameter ecology program where one studies problems under a plethora of parameters, and recently there has 
been a lot of active research~\cite{FJR13,JRV14,Cai03a,guo2004structural,niedermeier2010reflections} in this area.  
In particular, identifying a parameter as small as possible, under which a problem is fixed-parameter tractable or has a polynomial-sized kernel, is an interesting research direction. 
We continue this line of research and consider parameterizations of {\sc Dominating Set} variants that are more 
natural and functions of the input graph.
\textit{Structural parameterization} of a problem is where the parameter is a function of the input structure rather than the standard output size. A recent trend in structural parameterization is to study a problem parameterized by deletion distance to various graph classes where the problem is trivial.
To the best of our knowledge, the preliminary version \cite{GJKMR18csr} of this paper is the first serious study of structural parameterization of any version of the dominating set problem, more particularly on deletion distance parameters.
After that, there have been some follow-up works on some other {\sc Dominating Set} variants, i.e. {\sc $(\rho, \sigma)$-Dominatng Set} \cite{AshokRTarxiv22}
	{\sc Locating Dominating Set} \cite{CappelleGS21Lagos,ChakrabortyFMT24} and {\sc Grouped Dominating Set} \cite{HanakaOOU23} with respect to structural parameters.

Our parameter of interest is the `distance' of the graph from a natural class of graphs. Here by distance, we mean the number of vertices whose deletion results in the class of graphs. Note that if {\sc Dominating Set} is {\NP}-hard in a graph class, then it will continue to be {\NP}-hard even on graphs that are $k$ away from the class, even for constant $k$ (in particular for $k=0$) and hence {\sc Dominating Set} is unlikely to be fixed-parameter tractable. Hence it is natural to consider graphs that are not far from a class of graphs where {\sc Dominating Set} is polynomial-time solvable. 
Our case study considers two such special graphs: cluster graphs, where each connected component is a clique, and split graphs, where the vertex set can be partitioned into a clique and an independent set. In the former, all the variants of dominating set we consider are polynomial-time solvable, while in the latter class of split graphs, we consider independent and efficient dominating set that are polynomial-time solvable.
We call the set of vertices whose deletion results in a cluster graph and split graph as \textit{cluster vertex deletion set} (CVD) and \textit{split vertex deletion set} (SVD), respectively.

Finally, we remark that the size of minimum CVD and minimum SVD is at most the size of a minimum vertex cover in a graph, which is a well-studied parameterization in the parameter-ecology program~\cite{FJR13}.
Each of the variants of {\sc Dominating Set} is FPT when the minimum vertex cover size is considered as the parameter.
Since the minimum CVD size and minimum SVD size are structurally smaller than the minimum vertex cover size, the efficiency of parameterized algorithms can change.

\subsection{Definitions, Our Results and Organization of the paper}
We start by describing the variants of dominating set we consider in the paper.
A subset $S \subseteq V(G)$ is a \textit{dominating set} if $N[S] = V(G)$. If $S$ is also an independent set, then $S$ is an \textit{independent dominating set}. If $S$ is also a clique, then $S$ is a \textit{dominating clique}.
It is called an \textit{efficient dominating set} if for every vertex $v \in V$, $|N[v] \cap S| = 1$. 
Note that a graph may not have an efficient dominating set (for example, for a $4$-cycle).
If for every vertex $v$, $|N(v)\cap S| \ge r$, $S$ is a \textit{threshold dominating set} with threshold $r$.
When $r=1$, $S$ is a \textit{total dominating set}. 
Note that for dominating set, the vertices in $S$ do not need other vertices to dominate them, but they do in a total dominating set. For more information on these dominating set variants, see~\cite{haynes1997domination}.
We will often denote the problems {\sc Dominating Set} by DS, {\sc Efficient Dominating Set} by EDS, {\sc  Independent Dominating Set} by IDS, {\sc Dominating Clique} by DC, {\sc Threshold Dominating Set} by {\sc ThDS}, and {\sc Total Dominating Set} by TDS in the rest of the article.

When we say that a graph $G$ is {\em $k$-away} from a graph in a graph class, we mean that there is a subset $S$ of $k$ vertices in the graph such that $G\setminus S$ belongs to the class. 

Now we describe the main results in the paper (See Table~\ref{table:nonlin} for a summary).
When parameterized by the deletion distance $k$ to cluster graphs, 
\begin{itemize}
\item
we can find a minimum dominating set in $3^k n^{\OO(1)}$ time. Within the same time, we can also find a minimum IDS or a minimum DC or a minimum EDS, or a minimum TDS. We also give an $\OO^*((r+2)^k)$ algorithm for a minimum threshold dominating set with threshold $r$. These algorithms are obtained through a dynamic programming approach for interesting generalizations of set cover, which may be of independent interest.
These results are discussed in Section~\ref{sec:cluster-upper-bounds}.
\item
We complement our upper bound results by showing that for dominating set, dominating clique, and total dominating set, $(2-\epsilon)^k n^{\OO(1)}$ algorithm is not possible for any $\epsilon > 0$ under what is known as Set Cover Conjecture. We also show that for IDS, $(2-\epsilon)^kn^{\OO(1)}$ algorithm is not possible for any $\epsilon > 0$ under the Strong Exponential Time Hypothesis (SETH) and for EDS no $2^{o(k)} n^{\OO(1)}$ algorithm is possible unless the Exponential Time Hypothesis (ETH) is false.
It also follows from our reductions that dominating set, dominating clique, total dominating set, and IDS do not have polynomial-sized kernels unless {\nka}.
These results are discussed in Section~\ref{sec:clusterlowerbounds}.
\end{itemize}
The standard dominating set and most of its variants are {\NP}-hard or \W[2]-hard in split graphs~\cite{raman2008short}. For the two variants IDS and EDS that are polynomial-time solvable in split graphs, we show that when parameterized by the deletion distance $k$ to split graphs,

\begin{itemize}
\item
IDS can be solved in $2^k n^{\OO(1)}$ time and provide an $\Omega(2^k)$ lower bound for any $\epsilon > 0$ assuming SETH.
We also show that IDS parameterized by $k$ has no polynomial kernel unless {\nka}.
\item
The $2^k$ barrier can be broken for EDS by designing an $3^{k/2}n^{\OO(1)}$ algorithm. This is one of the very few problems with a runtime better than $2^k n^{\OO(1)}$ in the realm of structural parameterization. We also show that no $2^{o(k)}n^{\OO(1)}$ algorithm is possible unless the ETH is false.
These results are discussed in Section~\ref{sec:eds-ids-svd}.
\end{itemize}

\begin{table}[t]
\begin{center}
\begin{tabular}{|c|c|c|c|c|}
\hline
    ~ & \multicolumn{2}{c|}{Cluster Deletion Set} & \multicolumn{2}{c|}{Split Deletion Set} \\
    \hline
    ~ & ~Algorithms~~ & ~~Lower Bounds~~ & ~~ Algorithms~ & ~~Lower Bounds\\
    \hline
    DS, DC, & $3^k n^{\OO(1)}$ $\star$ & $\Omega(2^k)$ and & -- & para-\NP-hard~\cite{raman2008short} \\
    TDS & ~ & No polynomial kernel $\star$ &  &~\\
    \hline
    IDS & $3^k n^{\OO(1)}$ $\star$ & $\Omega(2^k)$ and & $2^k n^{\OO(1)}$ $\star$ & $\Omega(2^k)$ and \\
     &  & No polynomial kernel $\star$ & ~ & No polynomial kernel $\star$\\
    \hline
    EDS & $3^k n^{\OO(1)}$ $\star$ & $2^{o(k)} n^{\OO(1)}$ $\star$  & $3^{k/2} n^{\OO(1)}$ $\star$ & $2^{o(k)} n^{\OO(1)}$ $\star$ \\
    \hline
    {\sc ThDS} & $(r+2)^k n^{\OO(1)}$ $\star$ & No polynomial kernel $\star$ & -- & para-{\NP}-hard~\cite{raman2008short} \\
\hline
\end{tabular}
\end{center}
\caption{Summary of results. See \ref{sec:problem-definitions} for problem definitions. Results marked $\star$ indicate our results.}
\label{table:nonlin}
\end{table}

\subsection{Related Work}
\label{sec:relatedwork}
Clique-width~\cite{courcelle2000upper} of a graph is a parameter that measures how close to a clique the graph is.
Courcelle et. al.~\cite{CMR2000} showed that for a graph with clique-width at most $k$, any problem expressible in $MSO_1$ (monadic second order logic of the first kind) has an FPT algorithm with $k$ as the parameter if a $k$-expression for the graph (a certificate showing that the clique-width of the graph is at most $k$) is also given as input. 
The clique-width of a graph that is $k$ away from a cluster graph can be shown to be $k+1$ (with a $k$-expression) and
all the dominating set variants discussed in the paper can be expressed in $MSO_1$ and hence can be solved in FPT time in such graphs. But the running time function $f(k)$ in Courcelle's theorem is huge (more than doubly exponential). 
Oum et al.~\cite{OSV2013} gave an $k^{O(k)} n^{\OO(1)}$ algorithm to solve the minimum dominating set for clique-width $k$ graphs without assuming that the $k$-expression is given.
There is a $4^{k} n^{\OO(1)}$ algorithm by Bodlaender et. al.~\cite{BLJRJV2010} for finding minimum dominating set in graphs with clique-width $k$ when the $k$-expression is given as input. It is easy to construct the $k$-expression for graphs $k$ away from a cluster graph and hence we have a $4^{k} n^{\OO(1)}$ algorithm. 
The algorithms we give in Section~\ref{sec:eds-ids-cvd}, not only improve the running time but also apply to other variants of dominating set.

\subsection{Techniques}
\label{sec:techniques}
The algorithms for {\sc Dominating Set} and its variants parameterized by deletion distance to cluster graphs involve reducing the problem to a variant of {\SetCover} where there is a partition of the family of sets. More specifically, in {\SCWP}, there is a universe $U$, a family $\FF$ of subsets of $U$, and a partition $\BB$ of $\FF$. The goal is to find a subfamily $\FF'$ such that it covers $U$, and from each part of the partition $\BB$, at least one set is picked. In the reduction, the universe corresponds to the modulator $S$ of the {\sc Dominating Set} instance, the family corresponds to the neighborhood of vertices in $V \setminus S$ in $S$, and the partition corresponds to the cliques in $G \setminus S$. Picking at least one set from each part implies that at least one vertex is picked in each clique, which would dominate the rest of the clique. The {\SCWP} problem is modified accordingly to solve other variants of dominating set.
The problems are solved via dynamic programming similar to the $2^{|U|} |\FF|^{\OO(1)}$ algorithm for {\SetCover} \cite{FJR13} with appropriate modifications to satisfy the additional partition requirement.

The $2^{k} n^{\OO(1)}$ algorithms for EDS and IDS parameterized by deletion distance to split graphs are obtained from simple observations. The subsequent faster algorithm for EDS with $1.732^{k} n^{\OO(1)}$ running time is obtained by carefully tailored reduction and branching rules for the problem. Finally, all the lower bounds are obtained by giving parameterized reductions from problems with pre-existing lower bounds, assuming well-known conjectures. 


\section{Preliminaries and Notations}
\label{sec:prelims}
We use $[n]$ to denote the set $\{1,\dotsc,n\}$. We use standard terminologies of graph theory book by Diestel~\cite{Diestel}.
For a graph $G = (V, E)$ we use $n$ to denote the number of vertices and $m$ to denote the number of edges.
For a vertex $v \in V(G)$, we denote by $N_G(v) = \{(u \in V(G) | (u,v) \in E(G)\}$ the
open neighborhood of $v$.
When the graph is clear from the context, we drop the subscript $G$.
We denote by $N[v]$ the close neighborhood of $v$, i.e. $N[v] = N(v) \cup \{v\}$.
For $S \subseteq V(G)$, we define $N(S) = \{v \in V(G) | \exists u \in S$ such that $(u,v) \in E(G)\} \setminus S$ and $N[S] = N(S) \cup S$.

For vertices $u, v \in V(G)$, we denote by $dist_G(u,v)$, the {\em distance} between $u$ and $v$ in $G$, which is the length of the shortest path between $u$ and $v$ in $G$.
By $N^{=2}(v)$ we denote the set of vertices $u$ such that $dist_G(u,v) = 2$. For $S \subseteq V(G)$, we define $N^{=2}(S)$ as the set of vertices $u$  such that there exists a vertex $v \in S$ with $dist_G(u,v) = 2$, and for all $v' \in S, v' \neq v$, $dist_G(u,v') \geq 2$. 
For $S \subseteq V(G)$, we use $G[S]$ to denote the subgraph induced by $S$.
We say that for vertices $u,v \in V$, $u$ \textit{dominates} $v$ if $v \in N(u)$.

\subsection*{Problem Definitions}
\label{sec:problem-definitions}
\defparproblem{{\SetCover}}{A universe $U$ and a family $\FF \subseteq 2^{U}$ and an integer $k$}{$|U|$}{Are there at most $k$ sets $A_1,\ldots,A_k \in \FF$ such that $\bigcup\limits_{i=1}^k A_i = U$?}

We say that a subfamily $\FF' \subseteq \FF$ {\em covers} a subset $W \subseteq U$ if for every element $w \in W$, there exists some set in $\FF'$ containing $w$. 

\defparproblem{{\sc CNF-SAT}}{A boolean formula $\phi$ in conjunctive normal form with $n$ variables and $m$ clauses}{$n$}{Is there an assignment which evaluates $\phi$ to true?}

\defparproblem{{\sc $3$-CNF-SAT}}{A boolean formula $\phi$ in conjunctive normal form with $n$ variables and $m$ clauses such that every clause has at most three literals.}{$n$}{Is there an assignment which evaluates $\phi$ to true?} 

We formally define the dominating set problem variants in graphs that are $k$ away from a cluster/split/empty graph.

\defparproblem{{\sc (DS/EDS/IDS/DC/TDS/ThDS)-CVD}}{An undirected graph $G = (V, E), S \subseteq V(G)$ such that every component of $G \setminus S$ is a clique and an integer $\ell$.}{$|S|$}{Is there a dominating set/efficient dominating set/independent dominating set/dominating clique/total dominating set/threshold dominating set in $G$ of size at most $\ell$?}

\defparproblem{{\sc (EDS/IDS)-SVD}}{An undirected graph $G = (V, E), S \subseteq V(G)$ such that $G \setminus S$ is a split graph and an integer $\ell$.}{$|S|$}{Does $G$ have an efficient dominating set/independent dominating set/dominating clique of size at most $\ell$?}


\defparproblem{{\sc (EDS/IDS)-VC}}{An undirected graph $G = (V, E), S \subseteq V(G)$ such that $S$ is a vertex cover and an integer $\ell$.}{$|S|$}
{Does $G$ have an efficient dominating set/independent dominating set of size at most $\ell$?}


We also define a problem called Maximum Minimal Vertex Cover (MMVC) parameterized by deletion distance to cluster/empty graph. 

\defparproblem{{\sc MMVC-(VC/CVD)}}{An undirected graph $G = (V, E), S \subseteq V(G)$ such that $S$ is a vertex cover/cluster vertex deletion set and an integer $\ell$.}{$|S|$}
{Does $G$ have a minimal vertex cover with at least $\ell$ vertices?}


\subsection*{Parameterized Complexity Notions}
\label{sec:fpt-definitions}
\begin{definition}[Fixed Parameter Tractability]
Let $L \subseteq \Sigma^* \times \mathbb{N}$ be a parameterized language. $L$ is said to be {\em fixed parameter tractable} (or {\FPT}) if there exists an algorithm $\BB$, a constant $c$ and a computable function $f$ such that for all $x$, for all $k$,  $\BB$ on input $(x,k)$ runs in at most $f(k)\cdot|x|^c$ time and outputs $(x,k) \in L$ if and only if $\BB([x,k]) = 1$.
We call the algorithm $\BB$ as {\em fixed parameter algorithm} (or {\FPT} algorithm).
\end{definition}

\begin{definition}[Parameterized Reduction]
\label{defn:parameterized-reduction}
Let $P_1, P_2 \in \Sigma^* \times \mathbb{N}$ be two parameterized languages.
Suppose there exists an algorithm $\BB$ that takes input $(x,k)$ (an instance of $P_1$) and constructs an instance $(x',k')$ of $P_2$ such that the following conditions are satisfied.
\begin{itemize}
\item $(x,k)$ is a {\sc Yes-Instance} if and only if $(x',k')$ is a {\sc Yes-Instance}.
\item $k' \in f(k)$ for some function depending only on $k$.
\item Algorithm $\BB$ must run in $g(k) |x|^c$ time.
\end{itemize}
Then we say that there exists a {\em parameterized reduction} from $P_1$ to $P_2$.
\end{definition}

\noindent
{\bf \W-hierarchy:} In order to capture the parameterized languages being {\FPT} or not, the \W-hierarchy is defined as \FPT $\subseteq$ \W[1] $\subseteq \cdots \subseteq$ {\XP}.
It is believed that this subset relation is strict.
Hence a parameterized language that is hard for some complexity class above {\FPT} is unlikely to be {\FPT}.
The following Theorem \ref{thm:relation-hardness} gives the use of {\em parameterized reduction}.
If a parameterized language $L \subseteq \Sigma^* \times \mathbb{N}$ can be solved by an algorithm running in time $\OO(n^{f(k)})$, then we say $L \in \XP$.
In such a situation we also say that $L$ admits an {\XP} algorithm.

\begin{definition}[para-\NP-hardness]
\label{defn:para-np-hardness}
A parameterized language $L \subseteq \Sigma^* \times \mathbb{N}$ is called para-{\NP}-hard if it is \NP-hard for some constant value of the parameter.
\end{definition}
It is believed that a para-{\NP}-hard problem is not expected to admit an {\XP} algorithm as otherwise it will imply {\sf P} $=$ {\NP}.


\begin{theorem}
\label{thm:relation-hardness}
Let there be a parameterized reduction from parameterized problem $P_1$ to parameterized problem $P_2$.
Then if $P_2$ is fixed-parameter tractable so is $P_1$.
Equivalently if $P_1$ is \W[i]-hard for some $i \geq 1$, then so is $P_2$.
\end{theorem}

A closely related notion to fixed-parameter tractability is the notion of kernelization defined below.

\begin{definition}[Kernelization]
Let $L \subseteq \sum^* \times \mathbb{N}$ be a parameterized language. Kernelization is a procedure that replaces the input instance $(x,k)$ by a reduced instance $(x^\prime, k^\prime)$ such that
\begin{itemize}
\item $k^\prime \leq f(k)$, $\vert x^\prime \vert \leq g(k)$ for some functions $f, g$ depending only on $k$.
\item $(x, k) \in L$ if and only if $(x^\prime, k^\prime) \in L$.
\end{itemize} 
The reduction from $(x, k)$ to $(x^\prime, k^\prime)$ must be computable in $poly(\vert x \vert + k)$ time. If $g(k) = k^{\OO(1)}$ then we say that $L$ admits a {\em polynomial kernel}.
\end{definition}

It is well-known that a decidable parameterized problem is fixed-parameter tractable if and only if it has a kernel.
However, the kernel size could be exponential (or worse) in the parameter.
There is a hardness theory for problems having polynomial sized kernel.
Towards that, we define the notion of polynomial parameter transformation.

\begin{definition}[Polynomial parameter transformation (PPT)]
\label{pptdefinition}
Let $P_1$ and $P_2$ be two parameterized languages. We say that $P_1$ is polynomial parameter reducible to $P_2$ if there exists a polynomial time computable function (or algorithm) $f:\sum^* \times \mathbb{N} \rightarrow \sum^* \times \mathbb{N}$, a polynomial $p:\mathbb{N} \rightarrow \mathbb{N}$ such that
$(x,k) \in P_1$ if and only if $f(x, k) \in P_2$ and $k' \leq p(k)$ where $f((x,k)) = (x', k')$. We
call $f$ to be a polynomial parameter transformation from $P_1$ to $P_2$.
\end{definition}

The following proposition gives the use of the polynomial parameter transformation for obtaining kernels for one problem from another.

\begin{prop}[\cite{bodlaender2011kernel}]
\label{pptProperty}
Let $P, Q \subseteq \Sigma^* \times {\mathbb N}$ be two parameterized problems and assume that there exists a PPT  from $P$ to $Q$. Furthermore, assume that the classical version of $P$ is {\NP}-hard and $Q$ is in {\NP}. Then if $Q$ has a polynomial kernel then $P$ has a polynomial kernel.
\end{prop}

We use the following conjectures and theorems to prove some of our lower bounds.

\begin{conjecture}[Strong Exponential Time Hypothesis (SETH)~\cite{IPZ01}]
\label{thm:seth}
There is no $\epsilon > 0$ such that $\forall q \geq 3$, {\sc $q$-CNFSAT} can be solved in $(2-\epsilon)^{n} n^{\OO(1)}$ time where $n$ is the number of variables in input formula.
\end{conjecture}

\begin{conjecture}[Exponential Time Hypothesis (ETH)~\cite{IPZ01,IP01}]\label{thm:eth}
{\sc $3$-CNF-SAT} cannot be solved in $2^{o(n)} n^{\OO(1)}$ time where the input formula has $n$ variables and $m$ clauses.
\end{conjecture}

We have the following theorem giving an algorithm for {\SetCover} parameterized by the size of the universe.

\begin{theorem}[Theorem 3.10. in \cite{FK2010}]
\label{thm:set-cover-para-universe}
The {\SetCover} problem can be solved in $2^{n} (m+n)^{\OO(1)}$ time where $n$ is the size of the universe and $m$ is the size of the family of subsets of the universe.
\end{theorem}

We have the following conjecture corresponding to the above theorem.

\begin{conjecture}[Set Cover Conjecture (SCC)~\cite{cygan2016problems}]
\label{thm:scc}
There is no $\epsilon > 0$ such that {\sc SET COVER} can be solved in $(2-\epsilon)^{n} (m+n)^{\OO(1)}$ time where $n$ is the size of the universe and $m$ is the size of the family of subsets of the universe.
\end{conjecture}

\begin{theorem}[\cite{dom2014kernelization}]
\label{theorem:no-poly-kernel-set-cover}
{\SetCover} parameterized by the universe size does not admit any polynomial kernel unless {\nka}.
\end{theorem}

\begin{theorem}[\cite{FortnowS11}]
\label{theorem:no-poly-kernel-cnf-sat}
{\sc CNF-SAT} parameterized by the number of variables admits no polynomial kernel unless {\nka}.
\end{theorem}

We refer to \cite{CFKLMPPS15,DF95} for more details related to parameterized complexity, W-hardness, and kernelization.

\section{Dominating Set Variants parameterized by CVD Size}
\label{sec:eds-ids-cvd}
\subsection{Upper Bounds}
\label{sec:cluster-upper-bounds}

In cluster graphs, a dominating set simply picks an arbitrary vertex from each clique. This dominating set is also efficient and independent. If there is only one clique in the cluster graph, the dominating set is also a dominating clique. Otherwise, we can conclude that the graph has no dominating clique. For threshold dominating set with threshold $r$, if there is a clique with at most $r$ vertices, we can conclude there is no threshold dominating set with threshold $r$ as even after picking all the vertices of the clique, a vertex in the clique is has at most $r-1$ neighbors. If every clique is of size at least $r+1$, we arbitrarily pick $r+1$ vertices from every clique so that every vertex has $r$ neighbors excluding itself. When $r=1$, such a set is a total dominating set.

We can assume that a cluster vertex deletion set $S$ of size $k$ is given with the input graph $G$. If not, we can use the algorithm by Boral et al. \cite{BCKP16} that runs in $1.92^{k} n^{\OO(1)}$ time that either outputs a CVD set of size at most $k$ or says that no such set exists.

\subsubsection{Dominating Set}
\label{sec:dom-set-cluster-upper-bounds}

In this section, we give an FPT algorithm for {\sc DS-CVD}.


Our FPT algorithm for {\sc DS-CVD}  starts with making a guess $S'$ for the solution's intersection with $S$. We delete vertices in $N[S'] \cap S$ from $G$ as these vertices have already been dominated and are not part of the solution. We keep the vertices of $N[S'] \cap (V \setminus S)$ even though they are already dominated by $S'$, as they can potentially be used to cover the remaining vertices of $S$.

Let us denote the cliques in $G'= G \setminus S$ as $C_1, C_2, \dotsc , C_q$ where $q \le n-k$. We label the vertices of $G'$ as $v_1, v_2, \dotsc, v_{|V \setminus S|}$ such that the first $l_1$ of them belong to the clique $C_1$, the next $l_2$ of them belong to clique $C_2$ and so on for integers $l_1, l_2, \dotsc, l_q$.
Note that it could be that some cliques have all their vertices dominated by $S'$. We have the following observation.

\begin{observation}\label{observation:ds-partially-dominated-cique}
Let $C_i, i \in [q]$ be a clique component in $G'$ such that there exists at least one vertex in $V(C_i)$ that is not dominated by $S$. Then any dominating set $X$ of $G$ extending $S$ must contain at least one vertex from $V(C_i)$.
\end{observation}
\begin{proof}
Let $u \in V(C_i)$ be a vertex in clique component $C_i$ that is not dominated by $S'$. Since $X \cap S = S'$, $u$ has to be dominated by a vertex in $G \setminus S$. Since vertices in $G \setminus S$ outside $C_i$ do not have edges to $u$, it follows that $X$ should contain one vertex from $V(C_i)$ to dominate $u$.
\end{proof}

Thus, the clique components $C_i$ are in one of two cases. The first case is when $V(C_i)$ is dominated by $S'$, in which case, the solution dominating set need not contain any vertex from $C_i$. The other case is when there is at least one vertex in $V(C_i)$ is not dominated by $S'$, in which case, by Observation \ref{observation:ds-partially-dominated-cique}, the solution should contain one vertex from $V(C_i)$.
 
We are left with the problem of picking the minimum number of vertices from the cliques to dominate the vertices of the cliques
that are not yet dominated (by $S'$), and $S \setminus N[S']$, the vertices in $S$ that are not yet dominated. We abstract out the problem below.

\defparproblem{{\sc DS-disjointcluster}}{An undirected graph $G=(V,E)$, $S \subseteq V$ such that every connected component of $G \setminus S$ is a clique, a $(0,1)$ flag vector $\textbf{f} = (f_1,f_2, \dotsc, f_q)$ corresponding for the cliques $(C_1,\ldots,C_q)$ and an integer $\ell$.}{$|S|$}{Does there exist a subset $T \subseteq V \setminus S$ of size at most $\ell$, that dominates all vertices of $S$ and the vertices of all cliques $C_i$ with flags $f_i=1$?}

For an instance $(G, S, \ell)$, of {\sc DS-CVD} with cliques $C_1, \dotsc, C_q$ and guess $S'$, we construct an instance $(\hat{G}, \hat{S}, \textbf{f}, \hat{\ell})$ of {\sc DS-disjointcluster} where 
\begin{itemize}
\item $\hat{G}= G \setminus (N[S']\cap S)$,
\item  $\hat{S} = S  \setminus (N[S']\cap S)$,  
\item $\hat{\ell} = \ell - |S'|$, and
\item for all $i \in [q]$, $f_i=1$ if the clique $C_i$ is not dominated by $S'$ and $0$ otherwise.
\end{itemize}

Thus, solving {\sc DS-CVD} boils down to solving {\sc DS-disjointcluster}.

\paragraph*{\bf Formulation as a {\SetCover} variant:}

We now formulate the {\sc DS-disjointcluster} problem as a variant of {\SetCover}. Given an instance of $(G, S, \textbf{f},  \ell)$  {\sc DS-disjointcluster}, we construct an instance $(U, \FF)$ of {\SetCover}.
We define the universe $U$ as the set $S$.
For each vertex $v \in V \setminus S$, let $S_{v} = N(v) \cap S$.
We define the family of sets $\FF = \{S_{v} \mid v \in V \setminus S\}$. 
A {\SetCover} solution $\FF' \subseteq \FF$ for $(U,\FF)$ will cover all the elements of $S$.
In the graph $G$ of {\sc DS-disjointcluster} , the vertices corresponding to the sets in $\FF'$ will dominate all the vertices in $S$. 
However, {\sc DS-disjointcluster} has the additional requirement of dominating the vertices of every clique $C_i$ with $f_i=1$. 
To do so, at least one vertex has to be picked from every such clique. 

To capture this, we need to modify the {\SetCover} problem. We define for each clique $C_i, i \in [q],$, the collection of sets corresponding to the clique vertices $\BB_i = \{S_v : v \in C_i \}$. We call such collections as \textit{blocks}.
Note that the number of blocks and the number of cliques in $G \setminus S$ are the same.
We order the sets in each block in the order of the vertices $v_1,\ldots,v_{|V \setminus S|}$.
We have the following problem, which is a slight generalization of {\SetCover}. 

\defparproblem{{\SCWP}}{A universe $U$, a family of sets $\FF = \{S_{1}, \dotsc, S_{m}\}$, a partition $\BB=(\BB_1,\BB_2,\dotsc, \BB_q)$ of $\FF$, a $(0,1)$ vector $\hat{\textbf{f}} = (f_1,f_2, \dotsc, f_q)$ corresponding to each block in the partition $(\BB_1,\BB_2,\dotsc,\BB_q)$ and an integer $\ell$.}{$|U|= k$}{Does there exist a subset $\FF' \subseteq \FF$ of size $\ell$ covering $U$ and from each block $\BB_i$ with flags $f_i=1$, at least one set is in $\FF'$?}

Given an instance of $(G, S, \textbf{f}, \ell)$  {\sc DS-disjointcluster}, we construct an instance $(U, \FF, \BB, \hat{\textbf{f}}, \hat{\ell} )$ of {\SCWP} where
\begin{itemize}
\item $U = S$,
\item $\FF = \{S_{v} \mid v \in V \setminus S\}$, with $S_{v} = N(v) \cap S$,
\item $\BB_i = \{S_v \mid v \in C_i \}$, $i \in [q]$,
\item $\hat{\textbf{f}}=\textbf{f}$, and
\item $\hat{\ell} = \ell$.
\end{itemize}

We have the following lemma, which proves that {\sc DS-disjointcluster} and {\SCWP} are equivalent problems.

\begin{lemma}\label{lemma-ds-disjoint-eq-scwp}
$(G, S, \textbf{f}, \ell)$ is a YES-instance of {\sc DS-disjointcluster}, if and only if $(U, \FF, \BB, \hat{\textbf{f}}, \hat{\ell} )$ is a YES-instance of {\SCWP}. 
\end{lemma}

\begin{proof}
Suppose $(G, S, \textbf{f}, \ell)$ is a YES-instance of {\sc DS-disjointcluster}. Thus, there exists $T \subseteq V \setminus S$ of size $\ell$, that dominates all vertices of $S$ and all vertices of all cliques $C_i$ with flags $f_i=1$. Let $\FF' = \{S_v| v \in T\}$.  The subfamily $\FF'$ covers $U$ as $T$ dominates $S$.  Since $T$ dominates all vertices of all cliques $C_i$ with flags $f_i=1$, there is at least one vertex from each such clique in $T$. Thus, from each block $\BB_i$ with flags $f_i=1$, at least one set is contained in $\FF'$. Therefore, $(U, \FF, \BB, \hat{\textbf{f}}, \hat{\ell} )$ is a YES-instance of {\SCWP}.

Conversely, suppose that $(U, \FF, \BB, \hat{\textbf{f}}, \hat{\ell} )$ is a YES-instance of {\SCWP}. Thus, there exists a subfamily $\FF' \subseteq \FF$ of size $\ell$ that covers $U$ and from each block $\BB_i$ with flags $f_i=1$, at least one set is in $\FF'$. We define $T = \{v | S_v \in \FF'  \}$. Clearly, $T \subseteq V \setminus S$. Since, $\FF'$ covers $U$, $T$ dominates $S$. Since each block $\BB_i$ with flags $f_i=1$ has at least one set in $\FF'$, $T$ contains at least one vertex from each clique $C_i$ with flags $f_i=1$. Thus, $T$ dominates all vertices of all cliques $C_i$ with flags $f_i=1$. Therefore, $(G, S, \textbf{f}, \ell)$ is a YES-instance of {\sc DS-disjointcluster}.
\end{proof}

Thus, we now focus on solving {\SCWP}.

\begin{lemma}
\label{lemma:set-cover-partiton-fpt}
{\SCWP} can be solved in $2^{|U|} (m+|U|)^{\OO(1)}$ time.
\end{lemma}

\begin{proof}
We give a dynamic programming algorithm to solve {\SCWP}.
The algorithm is similar to the algorithm to solve {\SetCover} parameterized by the number of elements~\cite{FK2010} but with some modifications to handle the blocks $\BB_i$ with flag $f_i=1$.
For every subset $W \subseteq U$, for every $j \in [m]$ and boolean flag $b \in \{0,1\}$, we define $OPT[W,j,b]$ as the size of the minimum sized subset $X$ of $\{S_{1}, \dotsc, S_{j}\}$ covering $W$ such that from each block $\BB_i$ with $f_i=1$, there is at least one set in $X$, except the block $\BB_x$ containing the set $S_{j}$ where we reset the flag to $b$ to indicate that at least $b$ sets are required to be in $X$ in that block.

Note that all the sets of the block $\BB_x$ may not be present in the instance corresponding to $OPT[W,j,b]$. Thus a solution for {\SCWP}  might have a set in $\BB_x$, which may not be in $\{S_{1}, \dotsc, S_{j}\}$ and therefore not present in the instance corresponding to $OPT[W,j,b]$. We define the flag $b \in \{0,1\}$ resetting the flag corresponding to $\BB_x$ to handle this case.

We have $OPT[W,1,b] = 1$ if $W \subseteq S_{1}$, else $OPT[W,1,b] = \infty$. To compute all the values of $OPT[W,j,b]$, we initially set all the remaining values to $\infty$ and give the following recursive formulation for $OPT[W,j+1,b]$ with $j \ge 1$.
 
\begin{itemize}
\item 
\underline{Case 1} : $S_{j+1}$ is not the first set in its block $\BB_x$.
\begin{equation*}
OPT[W, j+1, b] = \min \Big\{ OPT[W,j,b], 1 + OPT[W \setminus S_{j+1}, j,0] \Big\}
\end{equation*}

In computing $OPT[W,j+1,b]$ recursively, either $S_{j+1}$ is in the solution or not. If $S_{j+1}$ is not in the solution, then we are left to cover $W$ using the sets in $\{S_{1}, \dotsc, S_{j} \}$. Thus, the solution is $OPT[W,j,b]$, which is already computed. If $S_{j+1}$ is in the solution, then it is true that from block $\BB_x$, at least one set is in the solution. Since, $S_{j+1}$ is picked, we no longer care about the requirement of picking at least $b \in \{0,1\}$ sets in block $\BB_x$. We are left to cover $W \setminus S_{j+1}$ using the sets in $\{S_{1}, \dotsc, S_{j} \}$. The corresponding solution is stored in  $OPT[W \setminus S_{j+1}, j,0]$.

\item 
\underline{Case 2} : $S_{j+1}$ is the first set in its block $\BB_x$.
\begin{equation*}
OPT[W,j+1,b] = 
\begin{cases}
1 + OPT[W \setminus S_{j+1}, j,f_{x-1}] \text{ if } b=1 \\
\min \Big\{ OPT[W,j,f_{x-1}], 1 + OPT[W \setminus S_{j+1},j,f_{x-1}] \Big\}  \\ \hspace*{42mm} \text{ if } b=0 
\end{cases}
\end{equation*}


When $S_{j+1}$ is the first element of the block, and $b$ is $1$, we are forced to pick $S_{j+1}$ in the solution as $S_{j+1}$ is the only set in the block $\BB_{x}$ for the instance corresponding to $OPT[W, j+1, b]$. Since $S_j$ is in the previous block $\BB_{x-1}$, we recursively look at the optimal solution with the flag set to $f_{x-1}$, the flag corresponding to $\BB_{x-1}$ in $\hat{\textbf{f}}$. If $b=0$, the set $S_{j+1}$ could be in the solution or not. We take the minimum of the solutions corresponding to both such cases.
\end{itemize}
 
We compute the subproblems in increasing order of subsets $W \subseteq U$ and for each $W$, in increasing order of $j$ and $b$. The solution to the problem is computed at $OPT[U,m,f_x]$ where $f_x$ is the flag value for the final block $\BB_x$ containing $S_{m}$. 
The number of problems is $2^{|U|+1} \cdot m$, and for each subproblem, we take $\OO(|U|)$ time to compute the set difference. Hence the total running time is $\OO(2^{|U|} \cdot m \cdot |U|)$ which is $2^{|U|} (m+|U|)^{\OO(1)}$. 
\end{proof}

From Lemmas \ref{lemma-ds-disjoint-eq-scwp} and \ref{lemma:set-cover-partiton-fpt} , we have the following corollary.

\begin{corollary} 
{\sc DS-disjointcluster} can be solved in $2^{|S|} n^{\OO(1)}$ time.
\end{corollary}

Since solving {\sc DS-disjointcluster} leads to solving {\sc DS-CVD}, we have the following theorem.

\begin{theorem}
{\sc DS-CVD} can be solved in $3^{k} n^{\OO(1)}$ time.
\end{theorem}

\begin{proof}
Given an instance $(G, S, k)$ of {\sc DS-CVD}, for each guess $S' \subseteq S$ with $|S'| = i$, we construct a {\sc DS-disjointcluster} instance $(\hat{G}, \hat{S}, f, \hat{\ell})$ with $|\hat{S}| = k-i$ and solve it with running time $2^{k-i} n^{\OO(1)}$. Hence the total running time is $\sum\limits_{i=1}^{k} {k \choose i} 2^{k-i} n^{\OO(1)}$ which is $3^k n^{\OO(1)}$.
\end{proof}

We now show that, with some careful modifications to the above dynamic programming algorithm, we can obtain improved FPT algorithms for minimum EDS, IDS, DC, TDS, and ThDS when parameterized by the size of the cluster deletion set.


\subsubsection{Efficient Dominating Set}
\label{sec:eds-cluster-upper-bounds}

Recall that for a graph $G$, a set $S \subseteq V(G)$ is called an efficient dominating set if for every vertex $v \in V(G)$, $|N[v] \cap X| = 1$. 

Like in {\sc DS-CVD}, the algorithm for {\sc EDS-CVD} starts by making a guess $S'$ from the modulator $S$. Note that for $S'$ to be part of an EDS in $G$, the sets $N[v]$ for $v \in S'$ has to be disjoint as otherwise, some vertex in $N[S']$ is dominated twice.
Observe that the vertices in $N[S']$ are dominated exactly once. Also,  no vertex in $N(S')$ can be added to the solution as that would mean a vertex in $S'$ is dominated twice. Thus, it is safe to delete the vertices in $N[S']$ and we do so. 

Recall that for a subset $S \subseteq V(G)$, we define $N^{=2}(S)$ as the set of vertices $u$ such that there exists a vertex $v \in S$ with $dist_G(u,v) = 2$, and for all $v' \in S, v' \neq v$, $dist_G(u,v') \geq 2$. We have the following lemma for the set of vertices $N^{=2}(S')$.

\begin{lemma}\label{lemma:eds-distance-2-vertices}
For an efficient dominating set $X$ in $G$, let $S' = X \cap S$. Then $X \cap N^{=2}(S') = \emptyset$.
\end{lemma}
\begin{proof}
For vertices $u \in N^{=2}(S')$, $dist_G(u,s) =2$ for some vertex $s \in S'$. This implies that there is a path $P = s - v - u$ in $G$ for some vertex $v \in N(S')$. If $u \in X$, then $|N(v) \cap X| =2$, contradicting that $X$ is an efficient dominating set.
\end{proof}


We have the following lemma.
\begin{lemma}\label{lemma:eds-distance-2-clique}
For an efficient dominating set $X$ in $G$, let $S' = X \cap S$. There does not exist a clique component $C$ in $G \setminus (S \cup N(S'))$ such that all the vertices of $V(C) \subseteq N^{=2}(S')$.
\end{lemma}
\begin{proof}
Suppose that there is a clique component $C$ in $G \setminus (S \cup N(S'))$ such that $V(C) \subseteq N^{=2}(S')$. Then from Lemma \ref{lemma:eds-distance-2-vertices}, we have $V(C) \cap X = \emptyset$. Note that all the neighbors of the vertices in $C$ in $G$ are in $S$ and $X \cap S = S'$.  Since no vertex in $S'$ is adjacent to vertices in $C$, we can conclude that there is no vertex in $X$ that is neighbor to vertices in $C$. This contradicts that $X$ is an efficient dominating set in $G$.
\end{proof}
Whenever we encounter a guess $S'$ with a component in $G \setminus (S \cup N(S'))$ whose vertices are all in the set $N^{=2}(S')$, we move on to the next guess of the intersection of the solution with $S$. The correctness follows from Lemma \ref{lemma:eds-distance-2-clique}.


We now have the following lemma for the case where there is a component of $G \setminus (S \cup N(S'))$ with at least one vertex that is not in $N^{=2}(S')$.

\begin{lemma}\label{lemma:eds-distance-clique-solution}
For an efficient dominating set $X$ in $G$, let $S' = X \cap S$. Let $C$ be a clique component in $G \setminus (S \cup N(S'))$ such that  $V(C) \nsubseteq N^{=2}(S')$. Then $|X \cap V(C)| =1$ with the vertex $x \in X \cap V(C)$ such that $x \notin N^{=2}(S')$.
\end{lemma}
\begin{proof}
Note that none of the vertices in $C$ have a neighbor in $S$ as $C$ is a component in $G \setminus (S \cup N(S'))$. Thus, to dominate the vertices of $C$, $X$ should contain a vertex in $C$. Also, note that such a vertex cannot be in $N^{=2}(S')$ from Lemma \ref{lemma:eds-distance-2-vertices}. The lemma follows.
\end{proof}

Note that in the above lemma, it could be that $V(C) \cap N^{=2}(S') = \emptyset$ as well.
From Lemma \ref{lemma:eds-distance-2-vertices}, the vertices in $N^{=2}(S')$ are never part of an efficient dominating set.
From Lemma \ref{lemma:eds-distance-2-clique}, the vertices in $N^{=2}(S')$ are present in components $C$ of $G \setminus (S \cup N(S'))$ with $V(C) \setminus N^{=2}(S')$ is non-empty. Also, by Lemma \ref{lemma:eds-distance-clique-solution}, any efficient dominating set contains exactly a vertex in $V(C) \setminus N^{=2}(S')$, which would also dominate all the vertices in  $N^{=2}(S')$ exactly once. Thus, we can safely delete all the vertices in $N^{=2}(S')$.



The problem now boils down to finding a subset of vertices $D$ such that
\begin{itemize} 
	\item for each clique component $C$ in $G \setminus (S \cup N[S'])$, $|D \cap V(C)|=1$. 
	\item $D \cap S = \emptyset$, and
	\item the vertices of $D \cup S$ form an efficient dominating set in the graph $G$.
\end{itemize}

We formalize this as the following problem.

\defparproblem{{\sc EDS-disjointcluster}}{An undirected graph $G = (V, E), S \subseteq V(G)$ such that $G \setminus S$ is a cluster graph.}{$|S|$}{Is there an efficient dominating set in $G$ disjoint from $S$?}

If the clique components of the cluster graph $G \setminus S$ are $C_1, C_2, \dotsc, C_q$, observe that the efficient dominating set corresponding to the {\sc EDS-disjointcluster} $T \subseteq V(G) \setminus S$  is such that $|T \cap V(C_i)| =1$ for all cliques $C_i, i  \in [q]$.
Like in {\sc DS-CVD}, we formulate the {\sc EDS-disjointcluster} problem as a variant of {\sc Set cover with Partition} problem. Since the remaining vertices in $S$ are to be covered exactly once by the vertices in $V(G) \setminus S$,  we show that the problem is equivalent to a variant of {\sc Exact Set Cover} problem defined as follows.

\defparproblem{{\ESCWP}}{A universe $U$, a family of sets $\FF = \{S_{1}, \dotsc, S_{m} \}$ , and a partition $\BB=(\BB_1,\BB_2,\dotsc, \BB_q)$ of $\FF$.}{$|U|$}{Is there a subset $\FF' \subseteq \FF$ such that every $u \in U$ is covered by exactly one set in $\FF'$ and from each block $\BB_i$ exactly one set is picked?}


Given an instance $(G,S)$ of {\sc EDS-disjointcluster}, we can construct an instance $(U, \FF, \BB)$ of {\ESCWP} where.
\begin{itemize}
\item $U = S$, 
\item $\FF = \{S_v \mid v \in V(G) \setminus S\}$ where $S_v = N(v) \cap S$, and
\item $\BB = \{ \BB_i \mid i \in [q]\}$ where $\BB_i = \{S_v \mid v \in V(C_i)\}$ for cliques $C_i$.
\end{itemize}
 

We have the following lemma.

\begin{lemma}\label{lemma-eds-disjoint-eq-escwp}
$(G, S)$ is a YES-instance of {\sc EDS-disjointcluster} if and only if $(U, \FF, \BB)$ is a YES-instance of {\ESCWP}.
\end{lemma}
\begin{proof}
Suppose $(G, S)$ is a YES-instance of {\sc EDS-disjointcluster}.
Thus, there exists an efficient dominating set $T \subseteq V \setminus S$ such that for each $v \in S, |N(v) \cap T| =1$ and for each clique $C_i$, $|V(C_i) \cap T| =1$.
Let $\FF' = \{S_v \mid v \in T\}$.
Since for every $v \in S$, $|N(v) \cap T| =1$ and $U=S$, every element $u \in U$ is covered by exactly one set from $\FF'$.
Also, since, for each clique $C_i$, $|V(C_i) \cap T | =1$,  each block $\BB_i$ contains  exactly one set in $\FF'$.
Therefore, $(U, \FF, \BB)$ is a YES-instance of {\ESCWP}.
Conversely, suppose that $(U, \FF, \BB)$ is a YES-instance of {\ESCWP}.
Then, there exists a subfamily $\FF' \subseteq \FF$ such that every element in $U$ is covered exactly once by $\FF'$ and each block $\BB_i$ contains exactly one set in $\FF'$.
We define $T = \{v \mid S_v \in \FF'  \}$. Clearly, $T \subseteq V \setminus S$.
Since, $\FF'$ covers elements in $U$ exactly once, for each $v \in S, |N(v) \cap T| =1$.
Also, since  each block $\BB_i$ contains exactly one set is in $\FF'$, for each clique $C_i$, $|V(C_i) \cap T| =1$.
Thus, $T$ is an efficient dominating set, and therefore, $(G, S)$ is a YES-instance of {\sc EDS-disjointcluster}.
\end{proof}



We now give an algorithm to solve {\ESCWP}.
\begin{lemma}\label{lemma:exact-set-cover-partiton-fpt}
{\ESCWP} can be solved in $2^{|U|}\cdot (m+|U|)^{\OO(1)}$ time. 
\end{lemma}
\begin{proof}
We give a dynamic programming algorithm to solve {\ESCWP} in the same vein as {\SCWP}. For every non-empty subset $W \subseteq U$, integer $j \in [m]$ and flag $b = \{0,1\}$, we define $OPT[W, j, b]$ as the size of the minimum cardinality subset $X$ of $\{S_{1}, \dotsc, S_{j} \}$ such that 
\begin{itemize}
\item each element of $W$ is covered exactly once by $X$, and
\item from each block $\BB_i$, there is exactly one set in $X$, except the block $\BB_x$ containing the set $S_{j}$ where we reset the flag to $b$ to indicate that exactly $b$ sets are required to be in $X$ from $\BB_x$.
\end{itemize}


We have $OPT[W, 1, b] = 1$ if $W = S_{1}$ and $b=1$, else $OPT[W, 1, b] = \infty$. To compute all the values of $OPT[W, j, b]$,  we initially set all the remaining values to $\infty$ and give the following recursive formulation for $OPT[W, j+1, b]$ with $j \ge 1$. 
\begin{itemize}
\item 
\underline{Case 1} : $S_{j+1}$ is not the first set in its block $\BB_x$  
\begin{equation*}
OPT[W, j+1, b]=
\begin{cases}
OPT[W, j, b] \textrm{  if }S_{j+1} \nsubseteq W \textrm{  or } b=0 \\
min \Big\{ OPT[W, j, b], 1 + OPT[W \setminus S_{j+1}, j, 0] \Big\} \textrm{  otherwise}
\end{cases}
\end{equation*}
\item 
\underline{Case 2} : $S_{j+1}$ is the first set in its block $\BB_x$.
\begin{equation*}
OPT[W, j+1,  b] = 
\begin{cases}
\infty \textrm{  if }S_{j+1} \nsubseteq W  \textrm{  and } b=1 \\
1 + OPT[W \setminus S_{j+1}, j, 1] \textrm{  if }S_{j+1} \subseteq W  \textrm{  and } b=1 \\
OPT[W, j, 1] \textrm{  if } b=0 
\end{cases}
\end{equation*}
\end{itemize}
The idea is similar to the earlier recursive formula for {\SCWP}. Note how we add $S_{j+1}$ to the solution only if is a subset of $W$. Since in the right-hand side of the recurrence, we have a state corresponding to $W \setminus S_{j+1}$, this ensures that every element in the universe is covered exactly once.
We compute the subproblems in increasing order subsets $W \subseteq U$ and for each $W$, in increasing order of $j$ and $b$. The solution to the problem is computed at $OPT[U,m,1]$. The running time of {\ESCWP} is same as that of {\SCWP} which is $2^{|U|}\cdot (m+|U|)^{\OO(1)}$. 
\end{proof}

From Lemmas \ref{lemma-eds-disjoint-eq-escwp} and \ref{lemma:exact-set-cover-partiton-fpt} , we have the following corollary.

\begin{corollary} 
{\sc EDS-disjointcluster} can be solved in $2^{|S|} n^{\OO(1)}$ time.
\end{corollary}

Since solving {\sc EDS-disjointcluster} leads to solving {\sc EDS-CVD}, we have the following theorem.

\begin{theorem}
\label{lemma:eds-cvd-fpt}
{\sc EDS-CVD} can be solved in $3^{k} n^{\OO(1)}$ time.
\end{theorem}
\begin{proof}
Given an instance $(G, S, k)$ of {\sc EDS-CVD}, for each guess $S' \subseteq S$ with $|S'| = i$, we construct a {\sc EDS-disjointcluster} instance $(\hat{G}, \hat{S}, f, \hat{\ell})$ with $|\hat{S}| = k-i$ and solve it with running time $2^{k-i} n^{\OO(1)}$. If  {\sc EDS-disjointcluster} returns $\infty$ for all choices of $S'$, return NO as there is no EDS in the graph.  The total running time is $\sum\limits_{i=1}^{k} {k \choose i} 2^{k-i} n^{\OO(1)}$ which is $3^k n^{\OO(1)}$.
\end{proof}

\subsubsection{Independent Dominating Set}
\label{sec:ids-cluster-upper-bounds}

\begin{theorem}
\label{lemma:ids-cvd-fpt}
There is a $3^{k} n^{\OO(1)}$ algorithm to solve {\sc IDS-CVD}.
\end{theorem}

\begin{proof}
The idea remains almost the same as in {\sc DS-CVD}.
For the guess in the modulator $S' \subseteq S$, the graph $G[S']$ has to be independent. 
If $S'$ dominates all the vertices of a clique, we can delete the clique as one cannot pick any vertex from this clique while preserving independence.
 In the graph obtained after deleting $N[S'] \cap S$, we must pick exactly one vertex from each clique to dominate vertices in $S \setminus N[S']$. Hence the  {\sc IDS-CVD} instance can be reduced to the {\SCWP} problem instance with a slight modification where instead of \textit{at least} picking one set from each block, we pick \textit{exactly} one set. A similar dynamic programming algorithm can solve this problem giving us an overall running time of $3^k n^{\OO(1)}$. 
\end{proof}

\subsubsection{Dominating Clique}
\label{sec:dc-cluster-upper-bounds}

\begin{theorem}
\label{lemma:dom-clique-cvd-fpt}
There is a $3^{k} n^{\OO(1)}$ algorithm to solve {\sc DC-CVD}.
\end{theorem}

\begin{proof}
Again, the idea remains almost the same as in {\sc DS-CVD}.
 The guess in the modulator $S' \subseteq S$ has to be a clique. 
 In the graph obtained after deleting $N[S'] \cap S$, suppose there are two clique components, without loss of generality say $C_1$ and $C_2$ in $G-S$ with vertices remaining to be dominated. We claim that it is a NO-instance of {\sc DC-CVD}. Note that we have to pick at least one vertex from each of the clique components $C_1$ and $C_2$ to dominate the remaining vertices in the components.  But there is no edge between the picked vertices as they are in different components of $G-S$. 
 
 Hence, there is at most one component in $G-S$ with vertices remaining to be dominated. Suppose there is exactly one such component $C_i$ for $i \in m$. If $S \setminus N[S']$ (undominated vertices in $S$) is empty, then we arbitrarily pick any vertex in $C_i \cap N(S')$. Note that we can only pick vertices in $C_i \cap N(S')$ so that the solution is a clique. Otherwise, the problem boils down to picking vertices in $C_i \cap N(S')$ so that the undominated vertices of $C_i$ and $S \setminus N[S']$ are dominated.
 
 Suppose there is no component in $G-S$ with vertices remaining to be dominated. The problem boils down to guessing a clique component $C_i$ for $i \in m$ and picking vertices $C_i \cap N(S')$ so that the vertices in $S \setminus N[S']$ are dominated.
 
 In both cases, we can define a {\SetCover} instance with the universe as $S \setminus N[S']$ and sets and neighborhood of vertices in $C_i \cap N(S')$ as the sets in the family. We then use Theorem \ref{thm:set-cover-para-universe} to solve {\SetCover} in $2^{|S \setminus N[S']|}$ time, leading to overall time of $3^{k} n^{\OO(1)}$.
\end{proof}
\subsubsection{Total and Threshold Dominating Set}
\label{sec:tds-thds-cluster-upper-bounds}

Again, our FPT algorithm starts with making a guess $S'$ for the solution's intersection with $S$. We delete vertices in $N[S']$ that have $r$ neighbors in $G[S']$.
For the rest of the vertices $v$ in $G$, we associate a weight $w(v)$ denoting the remaining number of times the vertex must be dominated after $S'$ is added to the dominating set (to be more precise, the integer $r-|N(v) \cap S'|$).
Now we are left to solve the following problem (for each guess $S'$).

\defparproblem{{\sc ThDS-disjointcluster}}{An undirected graph $G = (V, E), S \subseteq V(G)$ such that $G \setminus S$ is a cluster graph, weight function $w :V \rightarrow [r] \cup \{0\}$ and an integer $\ell$.}{$|S|$}{Is there a subset $D \in V \setminus S$ of size $\ell$ in $G$ such that every vertex $v \in V$ has at least $w(v)$ neighbours in $D$?}

We now give an $(r+1)^{|S|} n^{\OO(1)}$ algorithm to solve {\sc ThDS-disjointcluster}.
For each clique component $C_i$, we have to pick at least $\max\limits_{v_j \in C_i}^{} w(v_j)+1$ vertices to dominate all the vertices in $C_i$ at least $r$ times. This can be viewed as the weight corresponding to the clique component $C_i$. 

We construct a variant of the {\SetCover} instance as done in {\sc DS-CVD}.
Since the elements of $U$ have to be covered multiple times, the problem that we reduce to is a weighted multicover problem. We have the following problem that also handles the additional covering requirement in each of the cliques of $G$. 

\defparproblem{{\WSMCWP}}{A universe $U$, a family of sets $\FF = \{S_{1}, \dotsc, S_{m} \}$, a partition $\BB=(\BB_1,\BB_2,\dotsc, \BB_q)$ of $\FF$, weight functions $w_U : U \rightarrow [r]$, $w_{\BB}: \BB \rightarrow [r] \cup \{0\}$ and an integer $\ell$.}{$|U|$}{Does there exist a subset $\FF' \subseteq \FF$ of size $\ell$ such that every $ u \in U$ is covered at least $w_U(u)$ times and from each block $\BB_i$ at least $w_{\BB}(\BB_i)$ sets are picked?}


\sloppy Let $(G,S, w, \ell)$ be an instance of {\sc ThDS-disjointcluster}.
We define an instance $(U, \FF, \BB, w_U, w_{\BB}, \ell)$ of {\WSMCWP} where
\begin{itemize}
\item $U = S$, 
\item $\FF = \{S_v | v \in V \setminus S\}$ where $S_v = N(v) \cap S$,  
\item $\BB = \{ \BB_i | i \in [q]\}$ for cliques $C_i$, 
\item  $w_U(u) = w(u)$ for $u \in U$, and
\item $w_\BB(\BB_i) = \max_{v \in V(C_i)} w(v)$ for  block $\BB_i$.
\end{itemize}
 
\begin{lemma}
$(G,S, w, \ell)$ is a YES-instance of {\sc ThDS-disjointcluster} if and only if $(U, \FF, \BB, w_U, w_\BB, \ell)$ is a YES-instance of {\WSMCWP}.
\end{lemma}
\begin{proof}
Suppose  $(G,S, w, \ell)$ is a YES-instance of {\sc ThDS-disjointcluster}.
Thus, there exist $T \subseteq V \setminus S$ such that for each $v \in S, |N(v) \cap T| =  w(v)$ and for each clique $C_i$, $|V(C_i) \cap T | = \max_{v \in V(C_i)} w(v)$. 
Let $\FF' = \{S_v| v \in T\}$.
Since for $v \in S, |N(v) \cap T| =  w(v)$ and $U=S$, every element $u \in U$ is covered exactly $w_U(u)$ times by $\FF'$. 
Also, since, for each clique $C_i$, $|V(C_i) \cap T | = \max_{v \in V(C_i)} w(v) = w_\BB(\BB_i)$,  each block $\BB_i$ has at least $w_{\BB}(\BB_i)$ sets in $\FF'$. Therefore, $(U, \FF, \BB, w_U, w_\BB, \ell)$ is a YES-instance of {\WSMCWP}. 
Conversely, suppose that $(U, \FF, \BB, w_U, w_\BB, \ell)$ is a YES-instance of {\WSMCWP}. 
Thus, there exists a subfamily $\FF' \subseteq \FF$ such that every element $u \in U$ is covered exactly $w_U(u)$ times by $\FF'$ and each block $\BB_i$ contains $w_\BB(\BB_i)$ sets in $\FF'$. We define $T = \{v | S_v \in \FF'  \}$. 
Clearly, $T \subseteq V \setminus S$. Since, $\FF'$ covers elements in $U$ $w_U(u)=w(v)$ times, for each $v \in S, |N(v) \cap T| =w(v)$. Also, since  each block $\BB_i$ contains  $w_\BB(\BB_i)$ sets  in $\FF'$, for each clique $C_i$, $|V(C_i) \cap T | =  \max_{v \in V(C_i)} w(v)$. 
Thus, $(G,S, w, \ell)$ is a YES-instance of {\sc ThDS-disjointcluster}.
\end{proof}
 
We now give a dynamic programming algorithm to solve {\WSMCWP}. 

\begin{lemma}
There is an algorithm to solve {\WSMCWP} with running time $(r+1)^{|U|} (m+|U|)^{\OO(1)}$.
\end{lemma}
\begin{proof}

For every $j \in [m]$, for every weight vector $w=(w_{u_1},w_{u_2}, \dotsc, w_{u_{|U|}})$ (where $u_i \in U$, $w_{u_i} \in [r] \cup {0}$) and flag $b = \{0,1, \dotsc ,r\}$, we define $OPT[j, w, b]$ as the cardinality of the minimum subset $X$ of $\{S_{1}, \dotsc, S_{j} \}$ such that each element $u_i \in U$ is covered at least $w_{u_i}$ times (note that $w_{u_i}=0$ indicates that there is no need to cover $u_i$) and from each block $\BB_i$, there is at least $w_{\BB}(\BB_i)$ sets in $X$ except the block $\BB_x$ containing the set $S_{j}$ where we reset the weight to $b$ to indicate that at least $f$ sets are required in that block.



We have $OPT[1, w, b] = 1$ if $S_{1}$ covers every element in $u_i \in U$ at least $w_{u_i}$ times and $b \le 1$, else $OPT[1, w, b] = \infty$. 
To compute all the values of $OPT[j, w, b]$, we initially set all the remaining values to $\infty$ and give the following recursive formulation for $OPT[j+1, w, b]$ with $j \ge 1$. 
\begin{itemize}
\item 
\underline{Case 1} : $S_{j+1}$ is not the first set in its block $\BB_x$. 
\begin{equation*}
\begin{split}
OPT[j+1, w, b] = \min \Big\{ OPT[j, w, b], 
 1 + OPT[j, w', \max\{b-1,0\}] \Big\}
\end{split}
\end{equation*}
where $w'$ is the weight-vector after subtracting $1$ from $w_{u_i}$ for each of the elements $u_i \in S_{j+1}$ where $w_{u_i} > 0$.
\item 
\underline{Case 2} :  $S_{j+1}$ is the first set in its block $\BB_x$.
\begin{equation*}
OPT[j+1, w, b] = 
\begin{cases}
1 + OPT[j, w', w_{\BB}(\BB_{x-1})] \text{ if } b=1 \\
min \Big\{ OPT[j, w, w_{\BB}(\BB_{x-1})],  
1 + OPT[j, w', w_{\BB}(\BB_{x-1})] \Big\}  \text{ if } b=0 \\
\infty \textrm{ otherwise }
\end{cases}
\end{equation*}

where $w'$ is the weight-vector after subtracting $1$ from $w_{u_i}$ for each of the elements $u_i \in S_{j+1}$ where $w_{u_i} > 0$.
\end{itemize}

Again the idea is similar to the recursive formulation of {\sc DS-CVD} with the choice of whether $S_{j+1}$ is in the optimal solution or not.
When it is picked, we decrease the weight requirements of the elements in the set by one and get the new weight vector $w'$.
When $S_{j+1}$ is the first set in the block, in the recursive subproblem, we set the flag to $w_{\BB}(\BB_{x-1})$ corresponding to the block $\BB_{x-1}$ that contains the set $S_j$.

We compute the subproblems in increasing order of $j$, $w$, and $b$. The solution to the problem is computed at $OPT[m,w_1,w_{\BB}(\BB_{x})]$ where $w_1$ corresponds to the weight-vector from the input function $w_U$ and $w_{\BB}(\BB_{x})$ is the weight of the block $\BB_x$ containing $S_{m}$.
The number of problems is $ m \cdot (r+1)^{|U|} \cdot (r+1)$ and for each subproblem, we take $\OO(|U|)$ time to update the weight vector.
Hence the total running time is $\OO(|\FF| \cdot (r+1)^{|U|} \cdot |U|) = (r+1)^{|U|} (m+|U|)^{\OO(1)}$. 
\end{proof}

\begin{corollary}
There is an algorithm to solve {\sc ThDS-disjointcluster} with running time $(r+1)^{|S|} n^{\OO(1)}$.
\end{corollary}

We have the following theorem.

\begin{theorem}
\label{lemma:thds-cvd-fpt}
There is a $(r+2)^{k} n^{\OO(1)}$ algorithm to solve {\sc ThDS-CVD}.
\end{theorem}
\begin{proof}
For each guess $S'$ in the modulator with $|S'| = i$, we get a {\sc ThDS-disjointcluster} instance with $|S| = k-i$ and solve it with running time $(r+1)^{k-i} n^{\OO(1)}$. 
Hence the total running time in solving {\sc ThDS-CVD} is $\sum\limits_{i=1}^{k} {k \choose i} (r+1)^{k-i} n^{\OO(1)} = (r+2)^k n^{\OO(1)}$. 
\end{proof}

Since {\sc TDS-CVD} is {\sc ThDS-CVD} with $r=1$, we have
\begin{corollary}
There is an $3^{k} n^{\OO(1)}$ algorithm to solve {\sc TDS-CVD}.
\end{corollary}

\subsection{Lower bounds}
\label{sec:clusterlowerbounds}
We first give lower bounds for {\sc DS-CVD}, {\sc DC-CVD} and {\sc TDS-CVD} by giving a reduction from the {\SetCover} problem.

We start with the following lemma.
\begin{lemma}\label{lemma:split-graph-domset-within-clique}
Let $G = (C \uplus I, E)$ be a connected split graph with clique $C$ and independent set $I$. For every dominating set $D$ in $G$ of size $k$, there exist a dominating set $D'$ of size at most $k$ such that $D' \cap I = \emptyset$.
\end{lemma}
\begin{proof}
First, note that there does not exist a vertex $u \in I$ with no neighbor in $C$. If so, note that $u$ is an isolated vertex in $G$ as $u$ cannot have neighbors in $I$ as well as $I$ is an independent set. This contradicts that $G$ is a connected graph.

Thus, for every $u \in I, N(u) \cap C \neq \emptyset$. We construct a set $D'$ from $D$ as follows. For every vertex $u \in D \cap I$, replace $u$ with an arbitrary vertex $v \in  N(u) \cap C$.  Note that now, $D' \cap I = \emptyset$. We claim that $D'$ is a dominating set in $G$. Suppose there exist a vertex $w$ that is not dominated by $D'$. Since $D$ is a dominating set and we only replaced vertices of $D \cap I$ in $D$, we can conclude that $w$ is dominated by a vertex $u \in D \cap I$. Thus, either $w = u$ or $w \in C$ as $I$ is an independent set. Let $v \in N(u) \cap C$ be the vertex that replaced $u$. Since $C$ is a clique, $v$ dominates all the vertices in $C$ including $w$ if $w \in C$. Also, $v$ dominates $u$ as $v \in N(u)$; eliminating the case that $w=u$ as well. Therefore, $D'$ is a dominating set in $G$.
\end{proof}

We have the following corollary.

\begin{corollary}\label{corollay:split-dc-tds-within-clique}
Let $G = (C \uplus I, E)$ be a connected split graph with clique $C$ and independent set $I$. For every dominating clique (or total dominating set) $D$ in $G$ of size $k$ ($|C|, k \geq 2$ for total dominating set), there exist a dominating clique (or total dominating set) $D'$ of size at most $k$ such that $D' \cap I = \emptyset$.
\end{corollary}
\begin{proof}
We construct the dominating set $D'$ as in Lemma \ref{lemma:split-graph-domset-within-clique}.  Since $D'$ is within the clique $C$, it is a dominating clique. Thus, the corollary follows for dominating clique.

If $|D'| \geq 2$, it is a total dominating set. This is because the vertices $v \in D'$ also have a neighbor $u \in N(v) \cap D'$ as $D'$ is a clique. If $|D'|=1$, we add an arbitrary vertex from $C \setminus D'$ to $D'$. From a similar argument, it follows that $D'$ is a total dominating set.
\end{proof}

In the following lemma, we prove the lower bounds for {\sc DS-CVD}, {\sc DC-CVD} and {\sc TDS-CVD} via a reduction from {\SetCover}. Though the reduction is a standard one, we give the proof for the sake of completeness, especially to clarify how the reduction works for {\sc Dominating Clique} and {\sc Total Dominating Set} problems.

\begin{lemma}
\label{lemma:general-about-ds-and-tds-redn}
There is a polynomial time algorithm that takes an instance $(U, \FF, \ell)$ of {\SetCover} where $\ell \geq 2$, and outputs an instance $(G, \ell)$ of {\sc DS-CVD} (or {\sc DC-CVD} or {\sc TDS-CVD}) where $G$ has a cluster vertex deletion set with $|U|$ vertices, and $(U,\FF,\ell)$ has a set cover of size $\ell$ if and only if $G$ has a dominating set (or a dominating clique, or a total dominating set) of size $\ell$.
\end{lemma}

\begin{proof}
For the {\SetCover} instance $(U,\FF, \ell)$, let  $U = \{ u_1, .....u_k\}$ and $\FF = \{S_1,...,S_m\}$. Let us construct the graph $G = (U \uplus V,E)$ with vertex sets $U = \{ u_1, .....u_k\}$ and $V = \{s_1,...,s_m\}$.
Thus, every vertex in $U$ corresponds to an element in the universe and every vertex in $V$ corresponds to a set in family $\FF$.
Add edges $(u_i,s_j)$ if and only if $u_i \in S_j$. Also, add all edges $(s_i,s_j), i \neq j$ to make $G[V]$ an $m$-clique (clique with $m$ vertices). Note that $G \setminus U$ is a cluster graph with $|U| = k$.

Now we prove that there is a subset $\FF' \subseteq \FF$ of size $\ell$ covering $U$ if and only if there is a dominating set (or a dominating clique, or a total dominating set) of size $\ell$ in $G$.

For the forward direction, we claim the vertices in $V$ corresponding to the subsets in $\FF'$ form a dominating set of size $\ell$.
More specifically, we claim the set $V' = \{s_i \in V|S_i \in \FF'\}$ is a dominating set.
Since $G[V]$ is a clique, by picking at least one vertex from $V$, all the vertices of $V$ are dominated. 
Since $\FF'$ covers all the elements of $U$, for each vertex $u_i \in U$, there exists a vertex $s_j \in V'$ such that $(u_i,s_j) \in E(G)$.
Hence all the vertices of $U$ are dominated by $V'$ as well. Hence $V'$ is a dominating set in $G$ of size $\ell$.  

For the reverse direction, we first note that $G$ is a split graph as $U$ is an independent set and $V$ is a clique. Thus, Lemma \ref{lemma:split-graph-domset-within-clique} implies there is a dominating set of size at most $\ell$ that does not contain any vertex from $U$.
Picking the sets corresponding to the vertices in $D \subseteq V$, we get a subset $\FF' \subseteq \FF$ of size at most $\ell$ covering the universe $U$.

We now prove the same claim for {\sc DC-CVD} or {\sc TDS-CVD}. Note that in the forward direction above, the set $V'$ is entirely within the clique in $G \setminus U$. Thus, $V'$ is a dominating clique. Since $|V'| = \ell \geq 2$, $V'$ is also a total dominating set. This is because the vertices $v \in V'$ also have a neighbor $u \in N(v) \cap V'$.

In the reverse direction, as $\ell \geq 2$, Corollary \ref{corollay:split-dc-tds-within-clique} implies that there is a dominating clique (or total dominating set) of size at most $\ell$ that does not contain any vertex from $U$. The corresponding sets in the family $\FF$ cover $U$.
\end{proof}


We have the following theorem which follows from Lemma \ref{lemma:general-about-ds-and-tds-redn} and Conjecture~\ref{thm:scc}. 

\begin{theorem}
\label{theorem:ds-cvd-lb}
{\sc DS-CVD}, {\sc DC-CVD} and {\sc TDS-CVD} cannot be solved in $(2- \epsilon)^k n^{\OO(1)}$ time for any $\epsilon > 0$ unless the Set Cover Conjecture fails. Furthermore, {\sc Dominating Set}, {\sc Dominating Clique} and {\sc Total Dominating Set} parameterized by deletion distance to even a single clique cannot be solved in $(2- \epsilon)^k n^{\OO(1)}$ time for any $\epsilon > 0$ unless the Set Cover Conjecture fails.

The second statement in the theorem follows by noting that in the reduction in Lemma \ref{lemma:general-about-ds-and-tds-redn}, the graph instance is such that deletion of $k$ vertices of the graph result in a single clique. 
\end{theorem}

\begin{proof}
Suppose that there is an algorithm solving {\sc DS-CVD} in $(2- \epsilon)^k n^{\OO(1)}$ time. 
Then by Lemma~\ref{lemma:general-about-ds-and-tds-redn}, we can solve {\sc Set-Cover} with $|U|=k$ in $(2- \epsilon)^k n^{\OO(1)}$ time violating the Set Cover Conjecture (SCC). 
This completes the proof. 
Note that the dominating set in the reduction is also total. Hence {\sc TDS-CVD} also cannot be solved in $(2- \epsilon)^k n^{\OO(1)}$ time for any $\epsilon > 0$ unless Set Cover Conjecture fails.
\end{proof}

The following theorem follows from Theorem~\ref{theorem:no-poly-kernel-set-cover} and Lemma~\ref{lemma:general-about-ds-and-tds-redn}.

\begin{theorem}
\label{theorem:ds-cvd-kernel-lb}
{\sc DS-CVD}, {\sc TDS-CVD} and {\sc ThDS-CVD} do not have polynomial sized kernels unless \nka.
\end{theorem}

\begin{proof}
From Theorem \ref{theorem:no-poly-kernel-set-cover}, we know that {\SetCover} parameterized by the size of the universe does not admit polynomial sized kernels unless {\nka}. Since the reduction provided in Lemma~\ref{lemma:general-about-ds-and-tds-redn} is a polynomial parameter transformation (PPT) (see Definition \ref{pptdefinition}), by Proposition \ref{pptProperty}, we have that these problems including {\sc ThDS-CVD} do not have polynomial kernel unless {\nka}.
\end{proof}

Note that the proof idea of Theorem~\ref{theorem:ds-cvd-kernel-lb} does not work for {\sc IDS-CVD}. 
This is because, in the reduction, we turn the graph induced by the family of sets into a clique. Hence, only one vertex of the clique can be in an independent dominating set of such a graph. Nevertheless, we give a
 $(2-\epsilon)^k n^{\OO(1)}$ lower bound for {\sc IDS-CVD} under SETH using the following lower bound result for a different problem which is {\sc MMVC-VC}.
Recall that {\sc MMVC-VC} refers to the problem of finding a maximum sized minimal vertex cover in a graph parameterized by the size of a given vertex cover. 

\begin{theorem}[\cite{Zehavi17}]\footnote{Note that the SETH based lower bound result and the result ruling out the existence of polynomial kernel from paper \cite{Zehavi17} use different constructions.}
\label{theorem:mmvc-cv-lb}
Unless SETH fails, {\sc MMVC-VC} cannot be solved in $(2-\epsilon)^k n^{\OO(1)}$ time.
Moreover, {\sc MMVC-VC} does not admit polynomial sized kernel unless \nka.
\end{theorem}

We now prove the following observation from which it follows that the complement of a maximum minimal vertex cover is a minimum independent dominating set. 
 This observation is already known \cite{BCP15,Zehavi17}.
But we still give proof for completeness.

\begin{observation}
\label{lem:relate-mmvc-ids}
If $T$ is a minimal vertex cover of the graph $G$, then $V(G) \setminus T$ is an independent dominating set in $G$.
Furthermore, if $T$ is a maximum minimal vertex cover, then $V(G) \setminus T$ is a minimum independent dominating set.
\end{observation}

\begin{proof}
A set $T \subseteq V(G)$ is a maximum minimal vertex cover when $T$ is a minimal vertex cover and among all minimal vertex covers, $T$ has the maximum number of vertices.
Let $D = V(G) \setminus T$.
Clearly, $D$ is an independent set.
Note that for all $u \in V(G)$, either $u \in D$ or if $u \notin D$, then $u \in T$.
Clearly, as $T$ is minimal, there must be a neighbor $v$ that is not in $T$.
Therefore, such a neighbor can only be in $D$.
Hence, $D$ is a dominating set implying that $D$ is an independent dominating set.

Now suppose that $D$ is not a minimum independent dominating set.
Then there exists another independent dominating set $D'$ such that $|D'| < |D|$.
Now consider $T' = V(G) \setminus D'$.
Clearly $|T'| > |S|$ and any vertex in $D'$ has some neighbor in $T'$ (otherwise $D'$ is not minimum).
Therefore, $T'$ is a minimal vertex cover.
But then $T$ is not a maximum minimal vertex cover and that is a contradiction.
Hence, $D$ which is $V(G) \setminus T$ is a minimum independent dominating set.
\end{proof}


From Observation \ref{lem:relate-mmvc-ids}, we know that for a solution $T$ of an {\sc MMVC-VC} instance, $V(G) \setminus T$ is a minimum independent dominating set, and therefore the solution of {\sc IDS-VC} of the same instance. 
Thus, the {\sc MMVC-VC} problem is equivalent to {\sc IDS-VC}. 
Due to Theorem \ref{theorem:mmvc-cv-lb}, we have the following result.

\begin{theorem}
\label{theorem:ids-vc-lb}
{\sc IDS-VC} cannot be solved in $(2-\epsilon)^k n^{\OO(1)}$ time unless {\sc SETH} fails. Moreover {\sc IDS-VC} does not have any polynomial kernel unless {\nka}.
\end{theorem}

We now note that any vertex cover is a cluster vertex deletion set. Hence the cluster vertex deletion set size parameter is at most the vertex cover size parameter. From Theorem \ref{theorem:ids-vc-lb}, we have the following result.

\begin{corollary}
\label{theorem:ids-cvd-lb}
{\sc IDS-CVD} cannot be solved in $(2- \epsilon)^k n^{\OO(1)}$ time for any $\epsilon > 0$ unless SETH fails.
Moreover {\sc IDS-CVD} does not have any polynomial kernel unless \nka.
\end{corollary}

For {\sc EDS-CVD}, we can only prove a weaker lower bound of $2^{o(k)}$ time assuming ETH, but we give the lower bound for EDS parameterized by even a larger parameter, i.e. the size of a vertex cover. We have the following results. 

\begin{theorem}
\label{thm:lower-bound-eds-vc-new}
{\sc EDS-VC} cannot be solved in $2^{o(|S|)}$ time unless the {\it ETH} fails.
\end{theorem}

\begin{proof}
A polynomial time reduction from {\sc $3$-SAT} to {\sc Exact-SAT} was given by Schaefer~\cite{Schaefer78}.
This reduction takes an instance of {\sc $3$-CNF-SAT} problem consisting of $n$ variables and $m$ clauses and reduces it into an instance of {\sc Exact-SAT} problem with $n + 6m$ variables and $5m$ clauses.
The input to the {\sc Exact-SAT} problem is a formula $\varphi$ and the objective is to check if there exists an assignment that satisfies exactly one literal from every clause.
After that, it is possible to reduce {\sc Exact-SAT} problem into {\sc EDS-VC} that goes similar to the line of the reduction described in Section~\ref{sec:seth-lower-bound-ids-svd}.

In this section, we independently give a direct reduction from {\sc 3-SAT} to {\sc EDS-VC} which we believe is a lot simpler.

{\bf Construction:} Let $\phi$ be the input formula with variables $x_1,\ldots,x_n$ and clauses $C_1,\ldots,C_m$.
For all $i \in [n]$, we create $v_i,\bar{v_i}$. 
The vertex $v_i$ corresponds to $x_i$ appearing in its pure form and $\bar{v_i}$ corresponds to $x_i$ appearing in negated form. 
Equivalently, the vertex $v_i$ corresponds to the literal $x_i$, and the vertex $\bar{v_i}$ corresponds to the literal $\neg x_i$.
For all $i \in [n]$, we add edge $(v_i,\bar{v_i})$.
We call this a variable gadget.

\begin{figure}[ht]
\centering
	\includegraphics[scale=0.24]{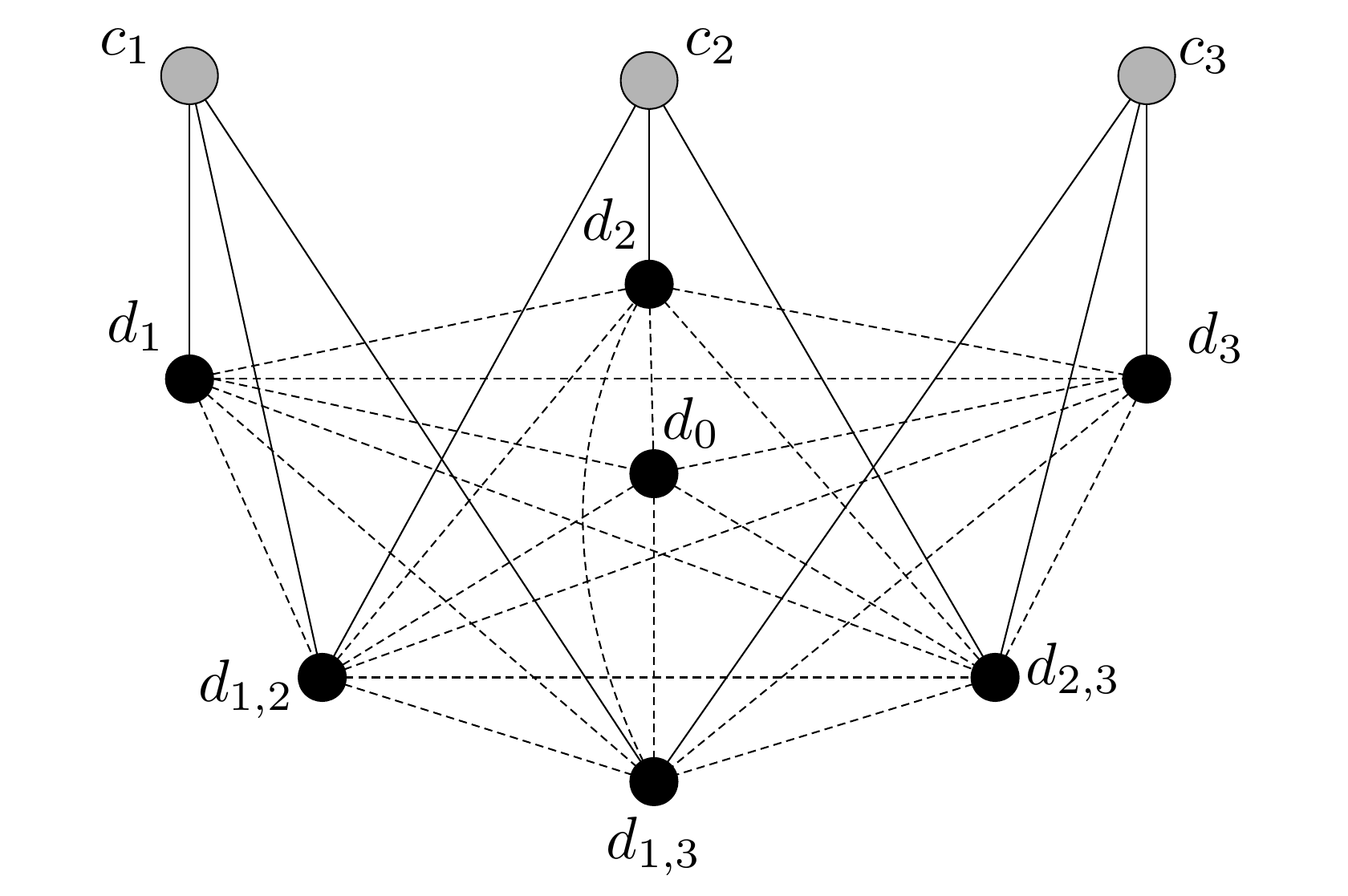}
	\caption{An illustration of the clause gadget described in Section~\ref{thm:lower-bound-eds-vc-new}. The vertices are $c_1,c_2,c_3,d_0,d_1,d_2,d_3,d_{1,2},d_{2,3},d_{1,3}$.
The edges are as follows: (i) $c_1$ is adjacent to $d_1, d_{1,2}$ and $d_{1,3}$, (ii) $c_2$ is adjacent to $d_2, d_{1,2}$ and $d_{2,3}$,
(iii) make $c_3$ adjacent to $d_3, d_{2, 3}$ and $d_{1,3}$, (iv) make $d_0$ adjacent to each of the six vertices $d_1, d_2, d_3, d_{1,2}, d_{2,3}$, and $d_{1,3}$, and (v) $\{d_0, d_1, d_2, d_3, d_{1,2}, d_{2,3}, d_{1,3}\}$ is a clique.}
\label{fig:3SAT-EDS-Vertex-Cover}
\end{figure}

Let $C$ be a clause. 
We create vertices $c_1,c_2,c_3,d_0,d_1,d_2,d_3,d_{1,2},d_{2,3},d_{1,3}$ and form a clause gadget as follows.
\begin{itemize}
	\item make $c_1$ adjacent to $d_1, d_{1,2}$ and $d_{1,3}$,
	\item make $c_2$ adjacent to $d_2, d_{1,2}$ and $d_{2,3}$,
	\item make $c_3$ adjacent to $d_3, d_{2, 3}$ and $d_{1,3}$,
	\item make $d_0$ adjacent to each of the six vertices $d_1, d_2, d_3, d_{1,2}, d_{2,3}$, and $d_{1,3}$, and
	\item finally, convert $\{d_0, d_1, d_2, d_3, d_{1,2}, d_{2,3}, d_{1,3}\}$ into a clique.
\end{itemize}
 
We refer to the Figure~\ref{fig:3SAT-EDS-Vertex-Cover} for an illustration of the clause gadget.
The vertices $c_1,c_2,c_3$ represent a copy of clause $C$.
Suppose a clause $C = (x_i \lor \neg {x_j} \lor x_p)$.
We add edges $(v_i,c_1),(\bar{v_j},c_2),(v_p,c_3)$.
This completes our construction.
Now we prove the following lemma.

\begin{lemma}
\label{lem:3sat-eds-vc-eqv}
$\phi$ is satisfiable if and only if $G_{\phi}$ has an EDS of size $n + m$.
\end{lemma}

\begin{proof}
$(\Rightarrow)$ Let $\phi$ be satisfiable.
Let $\bar{a} = (a_1,\ldots,a_n)$ be a satisfying assignment where $\forall i \in [n]: a_i \in \{true,false\}$.
We construct $D$ as follows.
If $a_i = true$, then we add $v_i$ into $D$.
Otherwise we add $\bar{v_i}$ into $D$.
Fix a clause $C$.
If only $c_1$ is dominated, then we pick $d_{2,3}$ into $D$.
The case is symmetric when only $c_2$ or $c_3$ is dominated.
Otherwise, if $c_1$ and $c_2$ are dominated but $c_3$ is not dominated, then we add $d_3$ into $D$.
Other cases are symmetric when exactly two from $c_1,c_2,c_3$ are dominated.
When all $c_1,c_2,c_3$ are dominated, then we pick $d_0$.
In this way, we pick exactly one vertex from each of the clause gadgets and get an efficient dominating set of size $n+m$.

$(\Leftarrow)$ Let $D$ be an efficient dominating set of size $n+m$.
Clearly, by construction exactly one vertex is picked from each of the clause gadgets and exactly one vertex is picked from each of the variable gadgets.
We construct an assignment $\bar{a} = (a_1,\ldots,a_n)$ as follows.
If $v_i \in D$, then we assign $a_i = true$, otherwise $\bar{v_i} \in D$ and we assign $a_i = false$.
Suppose $\bar{a}$ does not satisfy $\phi$.
Then there exists a clause $C$ that is not satisfied by $\bar{a}$.
It means that none of $c_1,c_2,c_3$ is dominated by variable vertices.
Now note that to dominate them at least two vertices are required from clause gadget.
But between every two vertices in clause gadget, the distance is two.
So, in such a case $D$ was not an efficient dominating set which is a contradiction.
So, $\bar{a}$ satisfies $\phi$ proving $\phi$ to be a satisfiable formula.
\end{proof}

From the above lemma and construction, we can construct a vertex cover of $G_{\phi}$ containing $2n + 6m$ vertices as follows.
Clearly, the vertices $\{v_i, \bar{v_i}\}$ for all $i \in [n]$ and the vertices $\{d_1,d_2,d_3,d_{1,2},d_{2,3},d_{1,3}\}$ for all clauses form a vertex cover with at most $2n + 6m$ vertices.
Hence, $G_{\phi}$ has a vertex cover having $2n + 6m$ vertices.
This provides a reduction from {\sc $3$-CNF-SAT} to {\sc EDS-VC}.

Suppose there exists an algorithm $\BB$ that solves {\sc EDS-VC} in $2^{o(k)} n^{\OO(1)}$ time.
Now we use algorithm $\BB$ to provide an algorithm for $3$-CNF-SAT running in $\OO(2^{o(n+m)})$ time.
Let $\phi$ be an instance of {\sc $3$-CNF-SAT} consisting of $n$ variables and $m$ clauses.
We construct $G_{\phi}$ as described above.
Now we use algorithm $\BB$.
If $\BB$ outputs $(G_{\phi},n+m)$ as ``NO'' instance then we output that $\phi$ is unsatisfiable.
Otherwise $\BB$ outputs an efficient dominating set of size $n+m$.
By construction as described in Lemma~\ref{lem:3sat-eds-vc-eqv}, we construct an assignment $\bar{a} = (a_1,\ldots,a_n)$ for {\sc $3$SAT}.
We know that $\bar{a}$ satisfies the clause by Lemma~\ref{lem:3sat-eds-vc-eqv}.
So we output $\bar{a}$ as a satisfying assignment.
As $\BB$ runs in $\OO(2^{o(n+m)})$ time and the transformation from $\phi$ to $G_{\phi}$ takes polynomial time, we get an algorithm for $3$-CNF-SAT in $\OO(2^{o(n+m})$ time.
It follows from~\cite{IPZ01,IP01} (See also Theorem $14.4$ of~\cite{CFKLMPPS15}) that unless ETH fails, $3$-CNF-SAT cannot be solved in $\OO(2^{o(n+m)})$ time.
This contradicts ETH.
So {\sc EDS-VC} cannot be solved in $2^{o(|S|)} n^{\OO(1)}$ time unless ETH fails.
It is well-known that any vertex cover is a cluster vertex deletion set.
So, {\sc EDS-CVD} also cannot be solved in $2^{o(k)} n^{\OO(1)}$ time unless ETH fails.
\end{proof}

Since the cluster vertex deletion set size parameter is at most the vertex cover size parameter for any graph, we have the following corollary.

\begin{corollary}
{\sc EDS-CVD} cannot be solved in $2^{o(|S|)} n^{\OO(1)}$ time unless {\it ETH} fails.
\end{corollary}

\section{Dominating Set variants parameterized by SVD size}
\label{sec:eds-ids-svd}
In this section, we address the parameterized complexity of dominating set variants when parameterized by the size of a given SVD set $S$.
Note that DS and TDS are \NP-hard on split graphs ~\cite{raman2008short}.  We have the following theorem.

\begin{theorem}
{\sc Dominating Clique} is {\NP-hard} on split graphs. 
\end{theorem}
\begin{proof}
From Lemma \ref{lemma:split-graph-domset-within-clique}, we can conclude that for every dominating set, there is a dominating set of the same size within the clique part of the split graph. Since this dominating set is within the clique, it is a dominating clique. 
Thus, an instance $(G,k)$ is a YES-instance of {\sc Dominating Set}  on split graphs if and only if $(G,k)$ is a YES-instance of {\sc Dominating Clique}  on split graphs. Since {\sc Dominating Set} is \NP-hard on split graphs, the theorem follows.
\end{proof}

Since DS, DC and TDS are {\NP-hard} on split graphs, they are para-\NP-hard when parameterized by the deletion distance to a split graph. Hence we focus only on EDS and IDS.

We assume that $S$ is given with the input. Otherwise given $(G, k)$, we use an 
$1.27^{k+o(k)} n^{\OO(1)}$ algorithm due to 
Cygan and Pilipczuk~\cite{CP13} to find a set of vertices of size at most $k$ whose removal makes $G$ a split graph.

\subsection{EDS and IDS parameterized by SVD size}

In this section, we illustrate the parameterized complexity of {\sc EDS-SVD} and {\sc IDS-SVD}. 
We restate the problem definitions.

\defparproblem{{\sc (EDS/IDS)-SVD}}{An undirected graph $G = (V, E), S \subseteq V(G)$ such that $G \setminus S$ is a split graph and an integer $\ell$.}{$|S|$}{Does $G$ have an efficient dominating set/independent dominating set/dominating clique of size at most $\ell$?}


\begin{theorem}
\label{thm:simple-eds-ids-svd-fpt}
{\sc EDS-SVD} and {\sc IDS-SVD} can be solved in $2^{|S|} n^{\OO(1)}$ time.
\end{theorem}

\begin{proof}
We will essentially prove that once we guess a correct subset $S'$ of $S$ that is in the EDS solution we seek, the remaining set of vertices can be determined in polynomial time, and then the claim will follow as we will try all possible guesses of $S'$.
As in the case of EDS parameterized by cluster vertex deletion set, once we guess the subset $S'$, we delete $N[S']$ and color
$N^2[S']$ red. 
Recall that we color a vertex by red if that vertex cannot be picked in a feasible efficient dominating set but has to be dominated.
Let $S'' = S \setminus N[S']$ be the remaining vertices of $S$ and $(C, I)$ be a partition of $G\setminus S$ into a clique $C$ and an independent set $I$.
To dominate the vertices of $C$, we need to pick some non-red vertex of $C$ or some non-red vertices from $I$. 
In particular, there are up to $|C|+1$ choices, i.e. either we pick exactly one vertex from $C$ or we decide not to pick any vertex from $C$ into our solution.
If we pick any vertex $u \in C$, we delete $N[v]$, and move the vertices of $N^{=2}(v)$ to $S''$. 
If we decide to pick no vertex from $C$, we move all the vertices of $C$ to $S''$.
After these choices have been made, all vertices of $C$ have been deleted (or moved to $S''$). 
Now if there are any red vertices in $I$, we move to the new guess, as such vertices cannot be dominated.
Otherwise, to dominate the remaining non-red vertices in $I$, we need to pick them all. Now we check whether the final solution picked is an EDS for the entire graph (in particular they should uniquely dominate $S'')$. This proves that {\sc EDS-SVD} can be solved in $\OO^*(2^{|S|})$ time.

The algorithm for IDS-SVD also works similarly.
First, we guess an independent set $S' \subseteq S$.
We delete $N[S']$ from $G$.
Now we are left with the split graph $(C, I)$ and vertices in $T = S \setminus N[S']$.
We have to use vertices from $C \cup I$ only to dominate vertices in $C \cup T \cup I$.
We guess vertices in $C$.
There can be at most $|C| + 1$ many guesses as at most one vertex can be an element of a solution.
If $v \in C$ is decided to be picked in the solution, then $N[v]$ is deleted.
Now $I \setminus N[v]$ is essential to be part of the solution.
If $A = S' \cup \{v\} \cup (I \setminus N[v])$ forms an independent dominating set, then we store $A$ as a candidate for being a solution.
If no vertex from $C$ is decided to be picked into the solution, then we have to pick all vertices from $I$ into the solution.
If $S' \cup I$ is a solution, then we store $S' \cup I$ also as a candidate for being a solution.
We go through all these candidates and choose one that is of the smallest cardinality.
We repeat this step for all possible subsets of $S$ that form an independent set.
So, for {\sc IDS-SVD} also, there exists an algorithm running in time $\OO^*(2^{|S|})$ for {\sc IDS-SVD}.
\end{proof}

\subsection{Improved Algorithm for EDS-SVD}
\label{sec:eds-split-deletion-set}

In this section, we give an improved algorithm for {\sc EDS-SVD} parameterized by the size of a given split vertex deletion set $S$ breaking the barrier of $2^k n^{\OO(1)}$.

Let $F = G \setminus S$.  
As $F$ is a split graph, $V(F) = C \uplus I$ where $C$ induces a clique and $I$ induces an independent set.
The algorithm uses a branching technique. Consider any efficient dominating set $D$ of a graph.
Any two vertices $u, v \in D$ must have a distance of at least three in $G$.
At any intermediate stage of the algorithm, we make a choice of not picking a vertex and we mark such vertices by coloring them red.
Other vertices are colored blue.
So, essentially, if a vertex is colored blue, we can freely pick it, and if a vertex is colored red, we are forbidden to pick it.

We initialize $D = \emptyset$ which is the solution set we seek.
Initially, we also color all the vertices by {\em blue} color.
Consider any pair of blue vertices $x, y \in S$.
If the distance between $x$ and $y$ is at most two in $G$, then we use the following branching rule. And we measure the progress of the algorithm by $\mu (G)$ which is the number of blue vertices in $S$, which is initially $|S| \leq k$.

\begin{branching rule}
\label{branch-rule:branch-on-pair-from-S}
Consider a pair of blue vertices $x, y \in S$ such that the distance between $x$ and $y$ is at most two in $G$.
We branch on the pair $(x,y)$ as follows.
\begin{enumerate}
    \item
    In the first branch, we add $x$ into $D$, delete $N[x]$ from $G$, color the vertices in $N^{=2}(x)$ by red.
    
    \item
    In the second branch, we add $y$ into $D$, delete $N[y]$ from $G$, color the vertices in $N^{=2}(y)$ by red.

    \item
    In the third branch, we color $x,y$ by red.
\end{enumerate}
\end{branching rule}

\begin{figure}[t]
    \centering
    \includegraphics[scale=0.3]{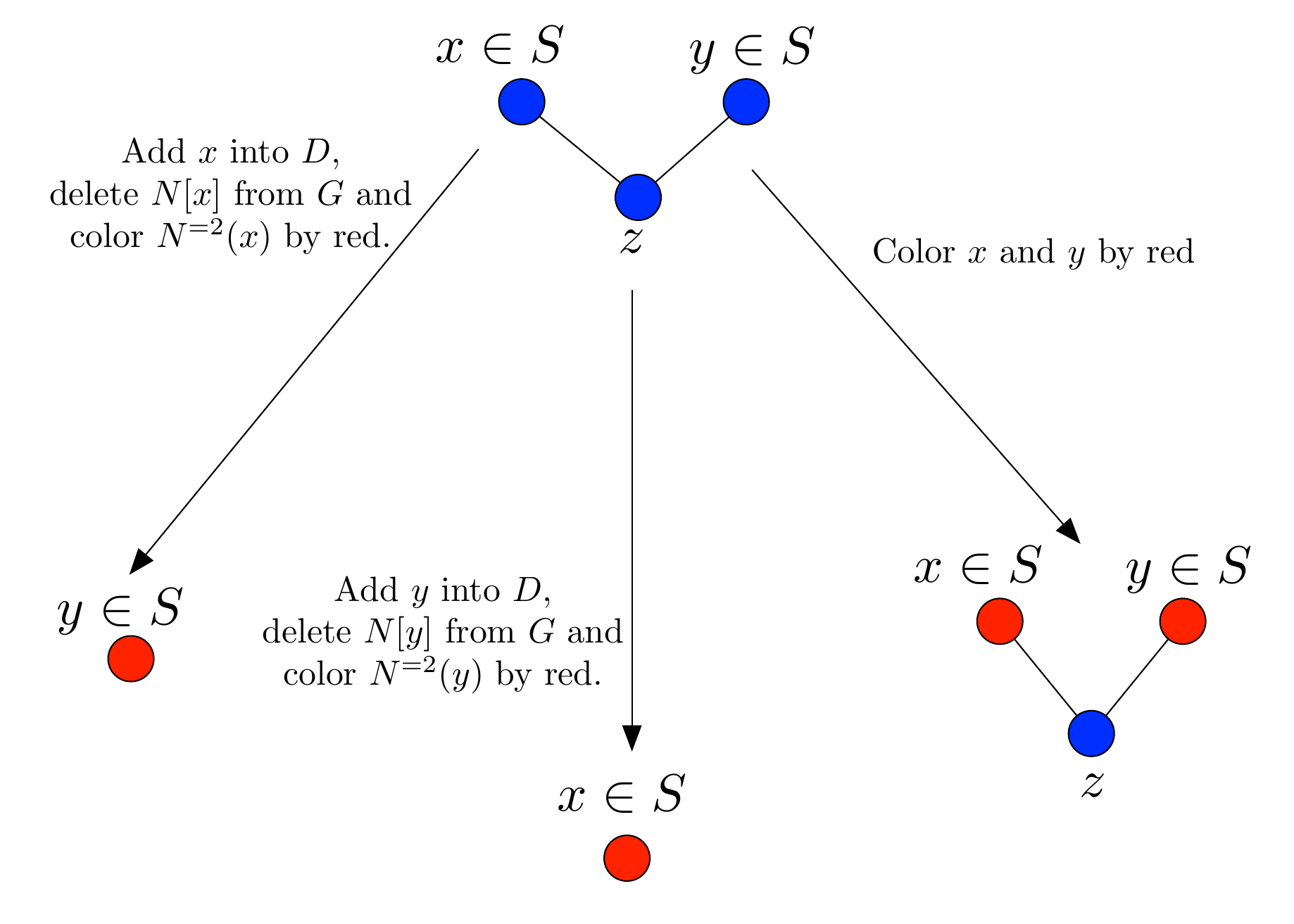}
    \caption{Illustration of Branching Rule \ref{branch-rule:branch-on-pair-from-S}. Note that the number of blue vertices drops by at least two in each of the branches.}
    \label{fig:eds-svd-branch-rule-on-pair-from-S}
\end{figure}

Clearly, the branches are exhaustive as both $x$ and $y$ cannot be in a feasible solution.
Furthermore, in the first branch, $x$ is deleted from $S$, and $y$ is colored red.
Symmetrically in the second branch, $y$ is deleted from $S$, and $x$ is colored red.
In the third branch, $x$ and $y$ are colored red. So in each of the branches, $\mu(G)$ drops by at least two resulting in a $(2,2,2)$ branch vector.
When this branching rule is not applicable, for every pair of blue vertices $x, y \in S$, $N[x] \cap N[y] = \emptyset$.
Since $C$ is a clique, we can have at most one vertex from $C$ in the solution.
When we decide to pick some vertex $v \in C$ into the solution, we delete $N[v]$ and color $N^{=2}(v)$ as red. 
So all vertices of $C$ get deleted.
 There are at most $|C|$ vertices in $C$.
When we decide not to pick any vertex from $C$ into the solution, then we color all vertices of $C$ as red.
Therefore, we have $(|C|+1)$ choices from the vertices of $C$. 
The measure $\mu(G)$ does not increase in any of these choices. 
A multiplicative factor of $(|C|+1)$ would come in the running time because of this `one-time' branching.
After having made our choices of vertex from $C$ in the solution, we are allowed to choose only some vertices from $I$.
We apply the following reduction rule to rule out some simple boundary conditions.

\begin{reduction rule}
\label{rule:unique-blue-neighbor}
If there exists a red vertex $x \in V(G)$ such that $N_G(x)$ has only one blue vertex $y$, then add $y$ into $D$, delete $N[y]$ from $G$ and color $N^{=2}(y)$ as {\it red}.
Also if there exists a blue vertex $x \in V(G)$ such that $N_G(x)$ contains no blue vertex, then add $x$ into $D$, delete $N[x]$ from $G$ and color $N^{=2}(x)$ as {\it red}.
\end{reduction rule}

It is easy to see that the above reduction rule is safe.
Note that we have some blue vertices in $I$.
Such vertices can only be dominated by themselves or a unique blue vertex in $S$, as otherwise Branching Rule~\ref{branch-rule:branch-on-pair-from-S} would have been applicable.
Now, suppose that there exists a blue vertex $x \in S$ that has at least two blue neighbors $u, v \in I$.
Our next reduction rule illustrates that we must pick $x$ into the solution.

\begin{reduction rule}
\label{rule:picking-blue-from-S}
If there exists a blue vertex $x \in S$ such that $N_G(x)$ contains at least two blue neighbors $u, v \in I$, then add $x$ into $D$, delete $N[x]$ from $G$ and color vertices in $N^{=2}(x)$ red.
\end{reduction rule}

If the above reduction rule is not applicable, we are ensured that a blue vertex in $S$ can have at most one blue neighbor in $I$.
In the following lemma, we give a formal proof why the above reduction rule is safe and the measure does not increase when the above reduction rule is applied.

\begin{lemma}
\label{lem:safeness-picking-blue-from-S}
Reduction Rule~\ref{rule:picking-blue-from-S} is safe.
\end{lemma}

\begin{proof}
The safety of this reduction rule is based on the fact that any feasible solution (if exists) must contain $x \in S$ under this construction.
Suppose that this is not the case.
We assume for the sake of contradiction that $D$ is an efficient dominating set that does not contain $x$.
Then either $u \in I$ or $v \in I$ but not both.
Note that any blue vertex in $I$ has only one blue neighbor in $S$, as all the vertices of $C$ are red.
So if $u \in D$ then $x, v \notin D$.
Then $v$ cannot be dominated at all.
Similarly if $v \in D$, then $u$ cannot be dominated at all.
So $x \in D$ and this concludes the proof.
\end{proof}

\begin{lemma}
\label{lem:no-measure-increase}
Reduction Rules~\ref{rule:unique-blue-neighbor} and~\ref{rule:picking-blue-from-S} can be implemented in polynomial time, and they do not increase $\mu(G)$.
\end{lemma}

\begin{proof}
It is not hard to see that both the reduction rules can be implemented in polynomial-time.
In order to give proof that $\mu(G)$ does not increase, our arguments are the following.
Reduction Rule~\ref{rule:unique-blue-neighbor} does not add any blue vertex into $S$, rather can delete a blue vertex from $S$.
Similarly, Reduction Rule~\ref{rule:picking-blue-from-S} only deletes a blue vertex from $S$.
So, $\mu$ does not increase in either of these two reduction rules.
\end{proof}

If there are red vertices in $I$ having no blue neighbor in $S$, then we move to the next branch as such a vertex cannot be dominated.
Thus any blue vertex in $I$ has only one blue neighbor in $S$ and any blue vertex in $S$ has only one blue neighbor in $I$.
As Reduction Rule~\ref{rule:unique-blue-neighbor} is not applicable, any red vertex $x \in S \cup C$ has at least two blue neighbors in $u, v \in N_G(x)$.
Clearly both $\{u,v\} \not\subset S$ as otherwise Branching Rule~\ref{branch-rule:branch-on-pair-from-S} would have been applicable.
So, now we are left with the case that $u,v \in I$ or $u \in I, v \in S$ but $(u,v)$ may or may not be an edge.
Now we apply the following branching rule.

\begin{branching rule}
\label{branch-rule:pair-of-non-adjacent}
Let $x$ be a red vertex in $S$ with two blue neighbors $u, v$.
\begin{enumerate}
	\item If $u, v \in I$, then we branch as follows.
	\begin{enumerate}
	    \item
	    In one branch we add $u$ into $D$, then delete $N[u]$ from $G$, we color $N^{=2}(u)$ as red. 
	   As $v \in N^{=2}(u)$ and $v$ has only one blue neighbor $z \in S$, we add $z$ also into $D$, delete $N[z]$ from $G$ and color $N^{=2}(z)$ by red.

        \item
        In the second branch, we add $v$ into $D$, delete $N[v]$ from $G$, color $N^{=2}(v)$ as red. 
        As $u \in N^{=2}(v)$ and $u$ has only one blue neighbor $y \in S$, we add $y$ also into $D$, delete $N^{=2}(y)$ from $G$ and color $N^{=2}(z)$ by red.

        \item
        In the third branch, color both $u$ and $v$ by red. Add the only blue neighbor $y$ of $u$ and $z$ of $v$ into $D$. Delete $N[y], N[z]$ from $G$ and color the vertices in $N^{=2}(y) \cup N^{=2}(z)$ by red.
	\end{enumerate}
	We refer to Figure \ref{fig:eds-svd-branch-rule-of-non-adjacent-part-one} for an illustration.

	\item $u \in I, v \in S, (u,v) \notin E(G)$, then we branch as follows.
	
	\begin{enumerate}
	    \item
	    In the first branch, we add $u$ to $D$ and color $v$ as red. This forces us to pick the only blue neighbor $z$ of $v$ where $z \in I$. So, we add $z$ to $D$. Delete  $N[u], N[z]$ from $G$ and color $N^{=2}(u), N^{=2}(z)$ as red.
    
       \item
       In the second branch, we color $u$ as red. This forces us to pick the only neighbor $y$ of $u$ where $y \in S$. And we pick $v$ into $D$ as well as $y$ into $D$. We delete $N[v], N[y]$ from $G$ and color $N^{=2}(v), N^{=2}(y)$ by red.

        \item
        In the third branch, we color both $u$ and $v$ by red. This forces us to pick the only blue neighbor $z \in N_G(v) \cap I, y \in N_G(u) \cap S$ into $D$. So, we pick $z$ into $D$, delete $N[z], N[y]$ from $G$ and color $N^{=2}(y), N^{=2}(z)$ by red.
	\end{enumerate}
	We refer to Figure \ref{fig:eds-svd-branch-rule-of-non-adjacent-part-two} for an illustration.
\end{enumerate}
\end{branching rule}

\begin{figure}
    \centering
    \includegraphics[scale=0.23]{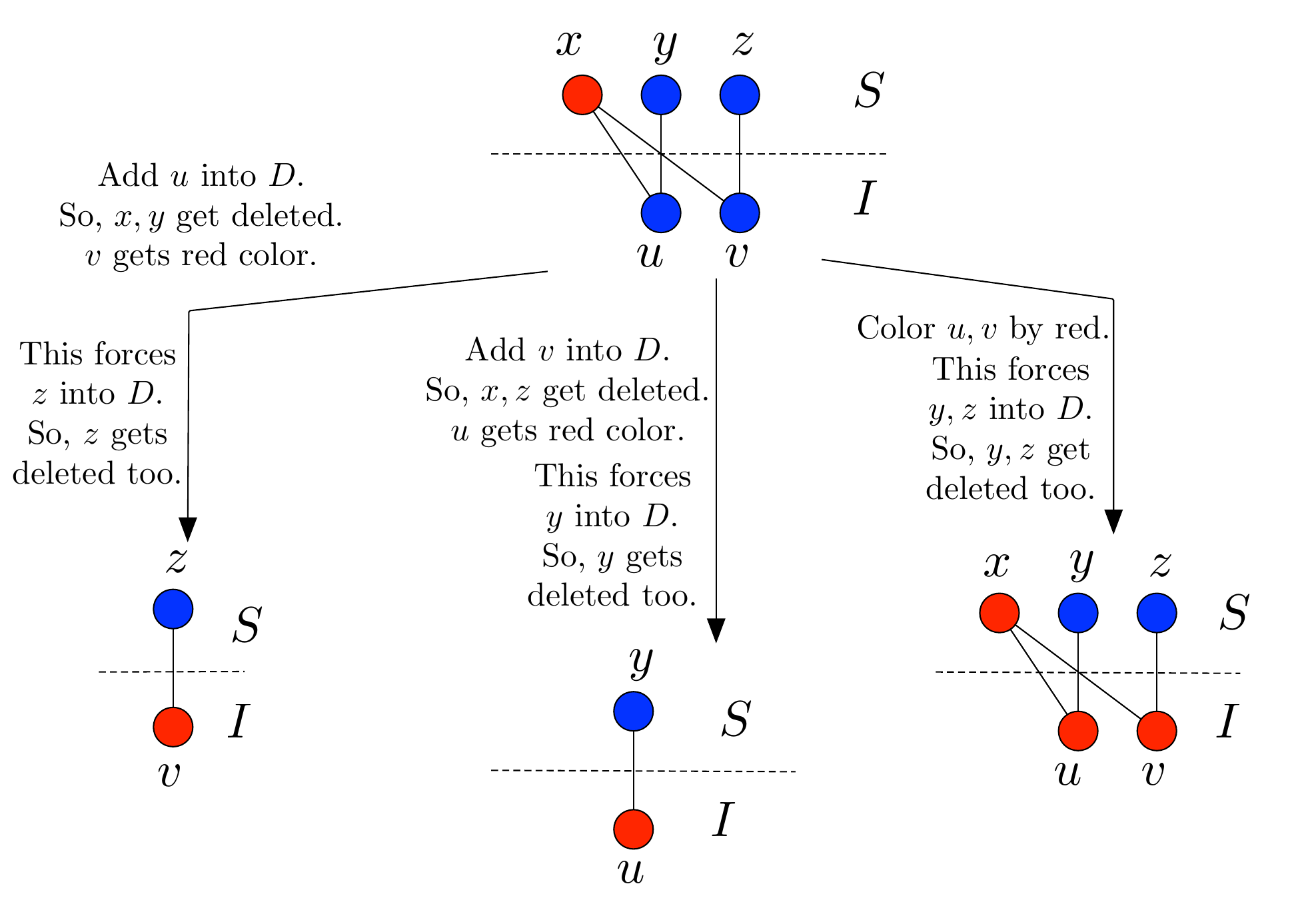}
    \caption{Illustration of Branching Rule~\ref{branch-rule:pair-of-non-adjacent} for the first case. Note that the number of blue vertices in $S$ drops by at least two in each of the branches.}
    \label{fig:eds-svd-branch-rule-of-non-adjacent-part-one}
\end{figure}

\begin{figure}
    \centering
    \includegraphics[scale=0.23]{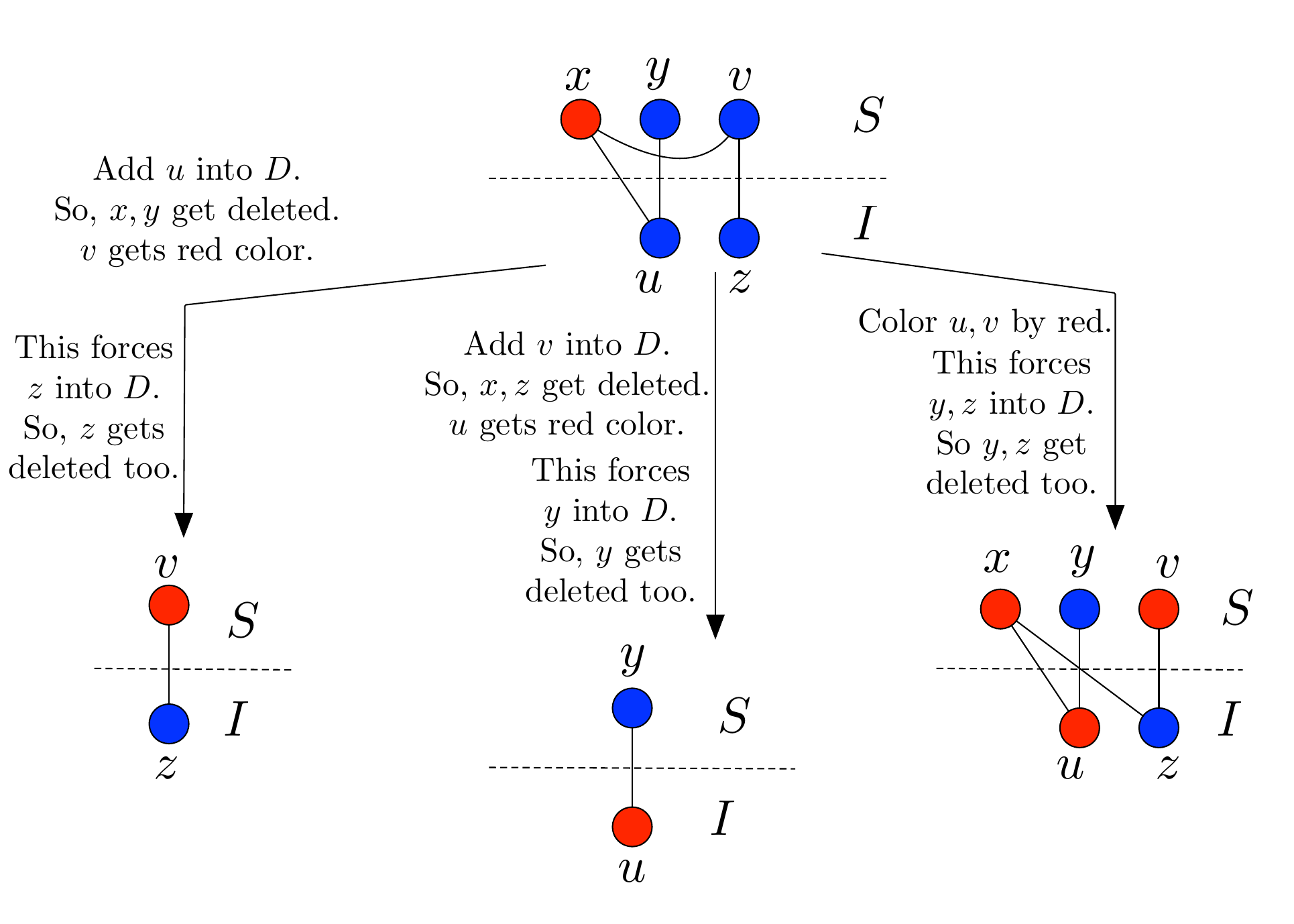}
    \caption{Illustration of Branching Rule~\ref{branch-rule:pair-of-non-adjacent} for the second case. Note that the number of blue vertices in $S$ drops by two in each of the branches.}
    \label{fig:eds-svd-branch-rule-of-non-adjacent-part-two}
\end{figure}

Observe that $\mu(G)$ drops by at least two in all three branches as eventually two blue vertices of $S$ get deleted in all the branches.
Hence, the branching vector is $(2, 2, 2)$.
If none of the above rules are applicable, then we have $u \in S, v \in I$ and $(u,v) \in E(G)$.
We know that either $u \in D$ or $v \in D$.
Consider the red vertices in $N(u)$ and red vertices in $N(v)$.
As Branching Rule~\ref{branch-rule:branch-on-pair-from-S}, Reduction Rule~\ref{rule:unique-blue-neighbor} and Branching Rule~\ref{branch-rule:pair-of-non-adjacent} are not applicable, by the following lemma, using which we can pick $u$ or $v$ arbitrarily.

\begin{lemma}
\label{lem:equal-neighbor-hood-property}
If Branching Rule~\ref{branch-rule:branch-on-pair-from-S}, Reduction Rule~\ref{rule:unique-blue-neighbor} and Branching Rule~\ref{branch-rule:pair-of-non-adjacent} are not applicable and $x$, then for any two arbitrary blue vertices $u, v$, it holds that $N(u) \setminus \{v\} = N(v) \setminus \{u\}$.
We refer to Figure \ref{fig:last-lemma-illustration} for an illustration.
\end{lemma}

\begin{proof}
Suppose that $x \in N(u) \setminus \{v\}$.
Clearly, $x$ is a red vertex by the premise.
As Branching Rule~\ref{branch-rule:branch-on-pair-from-S} is not applicable, $x$ cannot have any other neighbor which is a blue vertex of $S$.
As Reduction Rule~\ref{rule:unique-blue-neighbor} is not applicable, $x$ has another blue neighbor and let $y$ be that blue neighbor.
Clearly, $y \in I$.
If $(u,y) \notin E(G)$, then Branching Rule~\ref{branch-rule:pair-of-non-adjacent} is applicable.
So $(u,y) \in E(G)$ implying that $y = v$ (as $v$ is the only blue neighbor of $u$).
So, $x \in N(v) \setminus \{u\}$ implying that $N(u) \setminus \{v\} \subseteq N(v) \setminus \{u\}$.
Similarly we can prove that $N(v) \setminus \{u\} \subseteq N(u) \setminus \{v\}$.
This implies that $N(u) \setminus \{v\} = N(v) \setminus \{u\}$.
\end{proof}

\begin{figure}[t]
\centering	
	\includegraphics[scale=0.24]{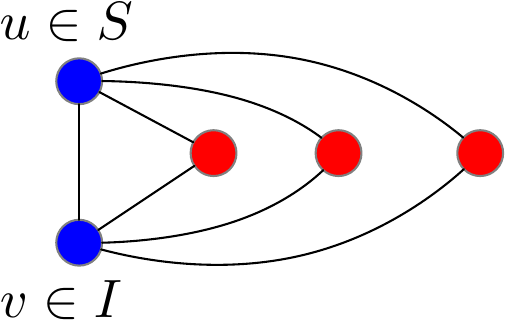}
	\caption{An illustration of Lemma \ref{lem:equal-neighbor-hood-property}. The red vertices are $N(v) \setminus \{u\} \subseteq S$ and $N(u) \setminus \{v\} \subseteq S$.}
\label{fig:last-lemma-illustration}
\end{figure}

This completes the description of our algorithm.
The measure is $k$ initially and the branching continues until $k$ drops to $0$.
So, we have the following recurrence.
$$T(k) \leq 3\cdot T(k-2) + \alpha(n+k)^c$$
Solving this recurrence, we get $1.732^k\cdot n^{\OO(1)}$ implying the following theorem.

\begin{theorem}
\label{thm:eds-svd-upper-bound}
{\sc EDS-SVD} can be solved in $3^{k/2} n^{\OO(1)}$ time.
\end{theorem}

\subsection{Lower Bounds for IDS and EDS}
\label{sec:seth-lower-bound-ids-svd}

It is well-known to us that the size of the SVD set is always smaller than the size of the vertex cover.
So, we have the following corollary as a consequence of Theorem \ref{theorem:ids-vc-lb}.

\begin{corollary}
{\sc IDS-SVD} cannot be solved in $(2-\epsilon)^k n^{\OO(1)}$ time unless {\sc SETH} fails and it does not admit polynomial kernels unless \nka.
\end{corollary}

Using similar arguments for EDS, we have the following corollary of Theorem~\ref{thm:lower-bound-eds-vc-new}.
\begin{corollary}
{\sc EDS-SVD} cannot be solved in $2^{o(|S|)} n^{\OO(1)}$ time unless {\it ETH} fails.
\end{corollary}

\section{Concluding Remarks}
Our paper provides a systematic study of the structural parameterizations of some of the {\sc Dominating Set} variants.
We have complemented with lower bounds based on ETH and SETH. 
One immediate open problem is to narrow the gap between upper and lower bounds, especially for the dominating set variants parameterized by the size of cluster vertex deletion set.
It has been observed that {IDS} is the complementary version of {\sc Maximum Minimal Vertex Cover} problem.
So a natural approach for an $2^k n^{\OO(1)}$ algorithm for {\sc IDS-CVD} is to apply the ideas used in~\cite{Zehavi17} to get $2^k n^{\OO(1)}$ algorithm for 
{\sc MMVC-VC}.
But this seems to require more work as there may not exist a minimal vertex cover that intersects the CVD set $S$ in a particular subset.

After the publication of the preliminary version of our paper, Misra and Rathi \cite{MisraR19} have studied {\sc Dominating Set} and {\EDS} and their respective generalizations to $r$-distance domination parameterized by deletion distance to edgless graphs, cluster graphs, split graphs, and complements of bipartite graphs. 
The authors improved the running-time of our algorithm for EDS parameterized by deletion distance to split graphs $k$ from $1.732^k n^{\OO(1)}$ to $1.619^k n^{\OO(1)}$. 
Bergougnoux et al.~\cite{BK2017} have given an $2^{O(k)} n^{\OO(1)}$ algorithm for connected dominating set (the dominating set induces a connected graph) for clique-width $k$ graphs when the $k$-expression is given as input. An interesting open problem is whether connected dominating set has a simpler FPT algorithm as in the FPT algorithms in this paper, when parameterized by the CVD set.

\paragraph*{\bf Dedication:} Rolf Niedermeier's work \cite{guo2004structural,niedermeier2010reflections} introduced the notion of studying problems parameterized by deletion distance to the class of graphs where the problem is in {\fP}, which our work is also part of. The last author fondly remembers his long association with Rolf Niedermeier as one of the early converts to parameterized complexity. His first meeting was in the first Dagstuhl on parameterized complexity and he was glad to know that Rolf was inspired by his early survey in the area. Since then both have had multiple meetings including a joint Indo-German project when Rolf was in Jena. He fondly remembers his interesting visits to Jena and Rolf's visit to Chennai. He is also particularly inspired by Rolf's constant quest for application areas where parameterized complexity can be applied.

\paragraph*{Acknowledgement:} Research of Diptapriyo Majumdar has been supported by Science and Engineering Research Board grant number SRG/2023/001592.

\bibliographystyle{plain}



\begin{thebibliography}{10}

\bibitem{alber2004polynomial}
Jochen Alber, Michael~R Fellows, and Rolf Niedermeier.
\newblock Polynomial-time data reduction for dominating set.
\newblock {\em Journal of the ACM (JACM)}, 51(3):363--384, 2004.

\bibitem{AshokRTarxiv22}
Pradeesha Ashok, Rajath Rao, and Avi Tomar.
\newblock Polynomial kernels for generalized domination problems.
\newblock {\em CoRR}, abs/2211.03365, 2022.

\bibitem{BK2017}
Benzamin Bergougnoux and Mamadou~M. Kant{\'e}.
\newblock Fast exact algorithms for some connectivity problems parametrized by
  clique-width.
\newblock {\em arXiv preprint arXiv:1707.03584}, 2017.

\bibitem{bodlaender2011kernel}
Hans~L. Bodlaender, Stephan Thomass{\'e}, and Anders Yeo.
\newblock Kernel bounds for disjoint cycles and disjoint paths.
\newblock {\em Theoretical Computer Science}, 412(35):4570--4578, 2011.

\bibitem{BLJRJV2010}
Hans~L. Bodlaender, Erik~Jan van Leeuwen, Johan M.~M. van Rooji, and Martin
  Vatshelle.
\newblock {Faster Algorithms on Branch and Clique Decompositions}.
\newblock In {\em MFCS}, pages 174--185. Springer, 2010.

\bibitem{BCKP16}
Anudhyan Boral, Marek Cygan, Tomasz Kociumaka, and Marcin Pilipczuk.
\newblock A fast branching algorithm for cluster vertex deletion.
\newblock {\em Theory Comput. Syst.}, 58(2):357--376, 2016.

\bibitem{BCP15}
Nicolas Boria, Fedarico~Della Croce, and Vangelis~Th. Paschos.
\newblock On the max min vertex cover problem.
\newblock {\em Discrete Applied Mathematics}, 196:62--71, 2015.

\bibitem{Cai03a}
Leizhen Cai.
\newblock {Parameterized Complexity of Vertex Colouring}.
\newblock {\em Discrete Applied Mathematics}, 127(3):415--429, 2003.

\bibitem{CappelleGS21Lagos}
M{\'{a}}rcia~R. Cappelle, Guilherme C.~M. Gomes, and Vin{\'{\i}}cius~Fernandes
  dos Santos.
\newblock Parameterized algorithms for locating-dominating sets.
\newblock In Carlos~E. Ferreira, Orlando Lee, and Fl{\'{a}}vio~Keidi Miyazawa,
  editors, {\em Proceedings of the {XI} Latin and American Algorithms, Graphs
  and Optimization Symposium, {LAGOS} 2021, Online Event / S{\~{a}}o Paulo,
  Brazil, May 2021}, volume 195 of {\em Procedia Computer Science}, pages
  68--76. Elsevier, 2021.

\bibitem{ChakrabortyFMT24}
Dipayan Chakraborty, Florent Foucaud, Diptapriyo Majumdar, and Prafullkumar
  Tale.
\newblock Tight (double) exponential bounds for identification problems:
  Locating-dominating set and test cover.
\newblock {\em CoRR}, abs/2402.08346, 2024.

\bibitem{ChlebikC06}
Miroslav Chleb{\'{\i}}k and Janka Chleb{\'{\i}}kov{\'{a}}.
\newblock Approximation hardness of edge dominating set problems.
\newblock {\em J. Comb. Optim.}, 11(3):279--290, 2006.

\bibitem{ChlebikC08}
Miroslav Chleb{\'{\i}}k and Janka Chleb{\'{\i}}kov{\'{a}}.
\newblock Approximation hardness of dominating set problems in bounded degree
  graphs.
\newblock {\em Inf. Comput.}, 206(11):1264--1275, 2008.

\bibitem{CMR2000}
Bruno Courcelle, Johann~A. Makowsky, and Udi Rotics.
\newblock Linear time solvable optimization problems on graphs of bounded
  clique-width.
\newblock {\em Theory of Computing Systems}, 33:125--150, 2000.

\bibitem{courcelle2000upper}
Bruno Courcelle and Stephan Olariu.
\newblock Upper bounds to the clique width of graphs.
\newblock {\em Discrete Applied Mathematics}, 101(1):77--114, 2000.

\bibitem{cygan2016problems}
Marek Cygan, Holger Dell, Daniel Lokshtanov, Daniel Marx, Jesper Nederlof,
  Yoshio Okamoto, Rammohan Paturi, Saket Saurabh, and Magnus Wahlstr{\"o}m.
\newblock {On problems as hard as CNF-SAT}.
\newblock {\em ACM Transactions on Algorithms (TALG)}, 12(3):41, 2016.

\bibitem{CFKLMPPS15}
Marek Cygan, Fedor~V. Fomin, Lukasz Kowalik, Daniel Lokshtanov, Daniel Marx,
  Marcin Pilipczuk, Michal Pilipczuk, and Saket Saurabh.
\newblock {\em Parameterized Algorithms}.
\newblock Springer, 2015.

\bibitem{CP13}
Marek Cygan and Marcin Pilipczuk.
\newblock {Split Vertex Deletion meets Vertex Cover: New fixed-parameter and
  exact exponential-time algorithms}.
\newblock {\em Inf. Process. Lett.}, 113(5-6):179--182, 2013.

\bibitem{Diestel}
Reinherd Diestel.
\newblock {\em Graph Theory}.
\newblock Springer, 2006.

\bibitem{dom2014kernelization}
Michael Dom, Daniel Lokshtanov, and Saket Saurabh.
\newblock {Kernelization lower bounds through colors and IDs}.
\newblock {\em ACM Transactions on Algorithms (TALG)}, 11(2):13, 2014.

\bibitem{downey1995fixed}
Rod~G Downey and Michael~R Fellows.
\newblock {Fixed-parameter tractability and completeness II: On completeness
  for W[1]}.
\newblock {\em Theoretical Computer Science}, 141(1-2):109--131, 1995.

\bibitem{DF95}
Rod~G. Downey and Michael~R. Fellows.
\newblock {Fixed-parameter tractability and completeness II: On completeness
  for W[1]}.
\newblock {\em Theoretical Computer Science}, 141(1-2):109--131, 1995.

\bibitem{downey2008parameterized}
Rodney~G Downey, Michael~R Fellows, Catherine McCartin, and Frances Rosamond.
\newblock Parameterized approximation of dominating set problems.
\newblock {\em Information Processing Letters}, 109(1):68--70, 2008.

\bibitem{FJR13}
Michael~R. Fellows, Bart M.~P. Jansen, and Frances~A. Rosamond.
\newblock Towards fully multivariate algorithmics: Parameter ecology and the
  deconstruction of computational complexity.
\newblock {\em European Journal of Combinatorics}, 34(3):541--566, 2013.

\bibitem{FK2010}
Fedor~V. Fomin and Dieter Kratsch.
\newblock {\em Exact Exponential Algorithms}.
\newblock Springer Science \& Business Media, 2010.

\bibitem{FortnowS11}
Lance Fortnow and Rahul Santhanam.
\newblock {Infeasibility of instance compression and succinct PCPs for {NP}}.
\newblock {\em J. Comput. Syst. Sci.}, 77(1):91--106, 2011.

\bibitem{garey1979computers}
M.~R. Garey and David~S. Johnson.
\newblock {\em Computers and Intractability: {A} Guide to the Theory of
  NP-Completeness}.
\newblock W. H. Freeman, 1979.

\bibitem{GJKMR18csr}
Dishant Goyal, Ashwin Jacob, Kaushtubh Kumar, Diptapriyo Majumdar, and
  Venkatesh Raman.
\newblock Structural parameterizations of dominating set variants.
\newblock In Fedor~V. Fomin and Vladimir~V. Podolskii, editors, {\em Computer
  Science - Theory and Applications - 13th International Computer Science
  Symposium in Russia, {CSR} 2018, Moscow, Russia, June 6-10, 2018,
  Proceedings}, volume 10846 of {\em Lecture Notes in Computer Science}, pages
  157--168. Springer, 2018.

\bibitem{guo2004structural}
Jiong Guo, Falk H{\"u}ffner, and Rolf Niedermeier.
\newblock A structural view on parameterizing problems: Distance from
  triviality.
\newblock In {\em Parameterized and Exact Computation: First International
  Workshop, IWPEC 2004, Bergen, Norway, September 14-17, 2004. Proceedings 1},
  pages 162--173. Springer, 2004.

\bibitem{HanakaOOU23}
Tesshu Hanaka, Hirotaka Ono, Yota Otachi, and Saeki Uda.
\newblock Grouped domination parameterized by vertex cover, twin cover, and
  beyond.
\newblock In Marios Mavronicolas, editor, {\em Algorithms and Complexity - 13th
  International Conference, {CIAC} 2023, Larnaca, Cyprus, June 13-16, 2023,
  Proceedings}, volume 13898 of {\em Lecture Notes in Computer Science}, pages
  263--277. Springer, 2023.

\bibitem{haynes1997domination}
Teresa~W. Haynes.
\newblock {\em Domination in graphs: advanced topics}.
\newblock Marcel Dekker, 2017.

\bibitem{haynes2023fundamentals}
Teresa~W Haynes, Stephen~T Hedetniemi, and Michael~A Henning.
\newblock Fundamentals of domination.
\newblock In {\em Domination in Graphs: Core Concepts}, pages 27--47. Springer,
  2023.

\bibitem{IP01}
Russel Impagliazzo and Rammohan Paturi.
\newblock On the complexity of k-sat.
\newblock {\em Journal of Computer and System Sciences}, 62:367--375, 2001.

\bibitem{IPZ01}
Russel Impagliazzo, Rammohan Paturi, and Francis Zane.
\newblock {Which Problems Have Strongly Exponential Complexity?}
\newblock {\em Journal of Computer and System Sciences}, 63(4):512--530, 2001.

\bibitem{JRV14}
Bart M.~P. Jansen, Venkatesh Raman, and Martin Vatshelle.
\newblock {Parameter Ecology for Feedback Vertex Set}.
\newblock {\em Tsinghua Science and Technology}, 19(4):387--409, 2014.

\bibitem{MisraR19}
Neeldhara Misra and Piyush Rathi.
\newblock The parameterized complexity of dominating set and friends revisited
  for structured graphs.
\newblock In Ren{\'{e}} van Bevern and Gregory Kucherov, editors, {\em Computer
  Science - Theory and Applications - 14th International Computer Science
  Symposium in Russia, {CSR} 2019, Novosibirsk, Russia, July 1-5, 2019,
  Proceedings}, volume 11532 of {\em Lecture Notes in Computer Science}, pages
  299--310. Springer, 2019.

\bibitem{niedermeier2010reflections}
Rolf Niedermeier.
\newblock Reflections on multivariate algorithmics and problem
  parameterization.
\newblock In {\em 27th International Symposium on Theoretical Aspects of
  Computer Science}. Schloss Dagstuhl-Leibniz-Zentrum fuer Informatik, 2010.

\bibitem{OSV2013}
Sang-il Oum, Sigve~Hortemo S{\ae}ther, and Martin Vatshelle.
\newblock {Faster Algorithms Parameterized by Clique-width}.
\newblock {\em arXiv preprint arXiv:1311.0224}, 2013.

\bibitem{philip2012polynomial}
Geevarghese Philip, Venkatesh Raman, and Somnath Sikdar.
\newblock Polynomial kernels for dominating set in graphs of bounded degeneracy
  and beyond.
\newblock {\em ACM Transactions on Algorithms (TALG)}, 9(1):1--23, 2012.

\bibitem{raman2008short}
Venkatesh Raman and Saket Saurabh.
\newblock {Short cycles make W-hard problems hard: FPT algorithms for W-hard
  problems in graphs with no short cycles}.
\newblock {\em Algorithmica}, 52(2):203--225, 2008.

\bibitem{Schaefer78}
Thomas~J. Schaefer.
\newblock {The Complexity of Satisfiability Problems}.
\newblock In {\em Proceedings of STOC}, pages 216--226, 1978.

\bibitem{van2011exact}
Johan M.~M. Van~Rooij and Hans~L Bodlaender.
\newblock Exact algorithms for dominating set.
\newblock {\em Discrete Applied Mathematics}, 159(17):2147--2164, 2011.

\bibitem{Zehavi17}
Meirav Zehavi.
\newblock {Maximum Minimal Vertex Cover Parameterized by Vertex Cover}.
\newblock {\em SIAM Journal of Discrete Mathematics}, 31(4):2440--2456, 2017.

\end{thebibliography}
\end{document}